\documentclass[11pt,twoside,a4paper]{article}

\def\DONOTINSERTCOMMENTS{}

\usepackage{ifthen}

\newboolean{conferenceversion}
\setboolean{conferenceversion}{false}

\ifthenelse{\boolean{conferenceversion}}
{}
{
  \usepackage[a4paper,margin=2.5cm]{geometry}
  \usepackage{tgtermes}
  \usepackage[scale=.92]{tgheros}
  \usepackage{tgcursor}
}

\usepackage{microtype}

\ifthenelse{\boolean{conferenceversion}}
{}
{
}

\ifthenelse{\boolean{conferenceversion}}
{}
{
  \usepackage[english]{babel}
}

\ifthenelse{\boolean{conferenceversion}}
{}
{
  \usepackage{fancyhdr}
}

\usepackage{subcaption} 

\usepackage{bussproofs}

\ifthenelse{\boolean{conferenceversion}}
{}
{
  
}

\usepackage{ifthen}

\usepackage{ifthen}

\provideboolean{detectedSTOC}
\provideboolean{detectedFOCS}
\provideboolean{detectedElsevier}
\provideboolean{detectedNOW}
\provideboolean{detectedLMCS}
\provideboolean{detectedIEEE}
\provideboolean{detectedPoster}
\provideboolean{detectedSIAM}
\provideboolean{detectedLNCS}
\provideboolean{detectedACM}
\provideboolean{detectedACMconf}
\provideboolean{detectedSigplanconf}
\provideboolean{detectedToC}
\provideboolean{detectedLIPIcs}
\provideboolean{detectedAAAI}
\provideboolean{detectedIJCAI}
\provideboolean{detectedCompCplx}
\provideboolean{detectedEasyChair}
\provideboolean{detectedArticle}
\provideboolean{detectedReport}
\provideboolean{detectedThesis}

\makeatletter

\@ifclassloaded{sig-alternate}
{\setboolean{detectedSTOC}{true}}
{\setboolean{detectedSTOC}{false}}

\@ifclassloaded{elsarticle}
{\setboolean{detectedElsevier}{true}}
{\setboolean{detectedElsevier}{false}}

\@ifclassloaded{now}
{\setboolean{detectedNOW}{true}}
{\setboolean{detectedNOW}{false}}

\@ifclassloaded{lmcs}
{\setboolean{detectedLMCS}{true}}
{\setboolean{detectedLMCS}{false}}

\@ifclassloaded{IEEEtran} {
  \setboolean{detectedIEEE}{true}
  \ifCLASSOPTIONconference {
    \setboolean{detectedFOCS}{true}
  }
  \else {
    \setboolean{detectedFOCS}{false}
  }
  \fi
}
{
  \setboolean{detectedFOCS}{false}
  \setboolean{detectedIEEE}{false}
}

\@ifclassloaded{siamart171218}
{\setboolean{detectedSIAM}{true}}
{\setboolean{detectedSIAM}{false}}

\@ifclassloaded{llncs}
{\setboolean{detectedLNCS}{true}}
{\setboolean{detectedLNCS}{false}}

\@ifclassloaded{acmsmall}
{\setboolean{detectedACM}{true}}
{\setboolean{detectedACM}{false}}

\@ifclassloaded{acmart}
{\setboolean{detectedACMconf}{true}}
{\setboolean{detectedACMconf}{false}}

\@ifclassloaded{sigplanconf}
{\setboolean{detectedSigplanconf}{true}}
{\setboolean{detectedSigplanconf}{false}}

\@ifclassloaded{toc}
{\setboolean{detectedToC}{true}}
{\setboolean{detectedToC}{false}}

\@ifclassloaded{lipics}
{\setboolean{detectedLIPIcs}{true}}
{\@ifclassloaded{lipics-v2019}
  {\setboolean{detectedLIPIcs}{true}}
  {\@ifclassloaded{oasics-v2019}
    {\setboolean{detectedLIPIcs}{true}}
    {\setboolean{detectedLIPIcs}{false}}
  }
}

\@ifclassloaded{cc}
{\setboolean{detectedCompCplx}{true}}
{\setboolean{detectedCompCplx}{false}}

\@ifclassloaded{easychair}
{\setboolean{detectedEasyChair}{true}}
{\setboolean{detectedEasyChair}{false}}

\@ifpackageloaded{aaai}
{\setboolean{detectedAAAI}{true}}        
{\@ifpackageloaded{aaai18}
  {\setboolean{detectedAAAI}{true}}
  {\@ifpackageloaded{aaai20}       
    {\setboolean{detectedAAAI}{true}}
    {\setboolean{detectedAAAI}{false}}}}

\@ifpackageloaded{ijcai18}
{\setboolean{detectedIJCAI}{true}}        
{\@ifpackageloaded{ijcai19}
  {\setboolean{detectedIJCAI}{true}}        
  {\setboolean{detectedIJCAI}{false}}}

\@ifclassloaded{sciposter}
{\setboolean{detectedPoster}{true}}
{\setboolean{detectedPoster}{false}}

\@ifclassloaded{article}
{\setboolean{detectedArticle}{true}}
{\setboolean{detectedArticle}{false}}

\makeatother

\ifthenelse{\not \isundefined{\examen} 
  \and \not \isundefined{\disputationsdatum} 
  \and \not \isundefined{\disputationslokal}}   
  {\setboolean{detectedThesis}{true}}
  {\setboolean{detectedThesis}{false}}

\ifthenelse{\boolean{detectedArticle} \or \boolean{detectedThesis}
  \or \boolean{detectedSTOC} \or \boolean{detectedFOCS}
  \or \boolean{detectedSIAM} \or \boolean{detectedIEEE}
  \or \boolean{detectedPoster}}
{\setboolean{detectedReport}{false}}
{\setboolean{detectedReport}{true}}

\ifthenelse{\boolean{detectedLIPIcs}}
{}
{\usepackage[utf8]{inputenc}}

\ifthenelse{\boolean{detectedIEEE}}
{\usepackage[cmex10]{amsmath}
  \interdisplaylinepenalty=2500}
{\usepackage{amsmath}}
\usepackage{amssymb}
\usepackage{amsfonts}

\usepackage{mathtools}

\usepackage{bussproofs}

\usepackage{microtype}

\usepackage{bm}  

\usepackage{subcaption} 

\ifthenelse
{     \boolean{detectedElsevier} \or \boolean{detectedSIAM}
  \or \boolean{detectedSIAM}     \or \boolean{detectedLIPIcs}}
{}
{\usepackage{varioref}}

\usepackage{xspace}

\ifthenelse    
{\boolean{detectedSTOC}      \or \boolean{detectedSIAM}         \or 
  \boolean{detectedLMCS}     \or \boolean{detectedNOW}          \or 
  \boolean{detectedLNCS}     \or \boolean{detectedACM}          \or
  \boolean{detectedToC}      \or \boolean{detectedLIPIcs}       \or
  \boolean{detectedACMconf}  \or \boolean{detectedAAAI}         \or
  \boolean{detectedCompCplx} \or \boolean{detectedSigplanconf}}
{}
{
  \usepackage{amsthm}
  \usepackage{amsthm-boldfont}
}

\ifthenelse        
{\boolean{detectedElsevier}  \or \boolean{detectedFOCS}         \or 
  \boolean{detectedSTOC}     \or \boolean{detectedPoster}       \or
  \boolean{detectedSIAM}     \or \boolean{detectedLMCS}         \or
  \boolean{detectedIEEE}     \or \boolean{detectedNOW}          \or
  \boolean{detectedToC}      \or \boolean{detectedThesis}       \or
  \boolean{detectedLNCS}     \or \boolean{detectedACM}          \or 
  \boolean{detectedLIPIcs}   \or \boolean{detectedAAAI}         \or
  \boolean{detectedACMconf}  \or \boolean{detectedIJCAI}        \or 
  \boolean{detectedCompCplx} \or \boolean{detectedSigplanconf}}
{}
{\usepackage{jncaptionsheadings}}

\usepackage{ifthen}
\usepackage{xspace}
\ifthenelse
{\boolean{detectedElsevier} \or \boolean{detectedSIAM} 
  \or \boolean{detectedLIPIcs}}
{}
{\usepackage{varioref}}

\DeclareMathAlphabet{\mathsfsl}{OT1}{cmss}{m}{sl}

\newcommand{\Bigoh}[1]{\mathrm{O} \bigl( #1 \bigr)}
\newcommand{\bigoh}[1]{\mathrm{O} ( #1 )}

\newcommand{\Bigomega}[1]{\Omega \bigl( #1 \bigr)}
\newcommand{\bigomega}[1]{\Omega ( #1 )}

\ifthenelse{\boolean{detectedToC}}{}
{

  \newcommand{\N}         {\mathbb{N}}
  \newcommand{\Nplus}     {\mathbb{N}^{+}}

}

\newcommand{\CEILING}[1]{\left \lceil #1 \right \rceil}

\newcommand{\MAXOFEXPR}[2][]{\max_{#1} \left\{ #2 \right\}}
\newcommand{\MINOFEXPR}[2][]{\min_{#1} \left\{ #2 \right\}}
\newcommand{\Maxofexpr}[2][]{\max_{#1} \bigl\{ #2 \bigr\}}
\newcommand{\Minofexpr}[2][]{\min_{#1} \bigl\{ #2 \bigr\}}
\newcommand{\maxofexpr}[2][]{\max_{#1} \{ #2 \}}

\newcommand{\MAXOFSET}[3][:]%
     {\ifthenelse{\equal{#1}{;}}%
     {\MAXOFEXPR{ #2 \,;\, #3 }}
     {\ifthenelse{\equal{#1}{:}}%
     {\MAXOFEXPR{ #2 \,:\, #3 }}
     {\max \twincommandJN{\left\{}{#2}{\left#1}{\right}{\,#3}{\right\}}}}}
\newcommand{\MINOFSET}[3][:]%
     {\ifthenelse{\equal{#1}{;}}%
     {\MINOFEXPR{ #2 \,;\, #3 }}
     {\ifthenelse{\equal{#1}{:}}%
     {\MINOFEXPR{ #2 \,:\, #3 }}
     {\min \twincommandJN{\left\{}{#2}{\left#1}{\right}{\,#3}{\right\}}}}}

\newcommand{\Maxofset}[3][:]%
     {\ifthenelse{\equal{#1}{;}}%
     {\Maxofexpr{ #2 \,;\, #3 }}
     {\ifthenelse{\equal{#1}{:}}%
     {\Maxofexpr{ #2 \,:\, #3 }}
     {\max \twincommandJN{\bigl\{}{#2}{\bigl#1}{\bigr}{\,#3}{\bigr\}}}}}
\newcommand{\Minofset}[3][:]%
     {\ifthenelse{\equal{#1}{;}}%
     {\Minofexpr{ #2 \,;\, #3 }}
     {\ifthenelse{\equal{#1}{:}}%
     {\Minofexpr{ #2 \,:\, #3 }}
     {\min \twincommandJN{\bigl\{}{#2}{\bigl#1}{\bigr}{\,#3}{\bigr\}}}}}

\newcommand{\F}{\mathbb{F}}

\newcommand{\gf}[1]{\mathrm{GF} ( #1 )}

\DeclareMathOperator{\Expop}{E}

\ifthenelse{\boolean{detectedLMCS}}
{}
{}

\newcommand{\twincommandJN}[6]%
    {#1#2#3\vphantom{#2#5}\mspace{-2.05mu}#4.#5#6}

\newcommand{\CondExp}[2]%
    {\Expop\twincommandJN{\bigl[}{#1}{\bigl|}{\bigr}{\,#2}{\bigr]}}
\newcommand{\CONDEXP}[2]%
     {\Expop\twincommandJN{\left[}{#1}{\left|}{\right}{\,#2}{\right]}}

\newcommand{\Condprob}[3][]%
    {\Pr_{#1}\twincommandJN{\bigl[}{#2}{\bigl|}{\bigr}{\,#3}{\bigr]}}
\newcommand{\CONDPROB}[3][]%
    {\Pr_{#1}\twincommandJN{\left[}{#2}{\left|}{\right}{\,#3}{\right]}}

\newcommand{\funcdescr}[3]{\ensuremath{ #1 : #2 \to #3}}

\newcommand{\domainof}[1]{\ensuremath{\mathrm{dom} ( #1 )}}

\newcommand{\set}[1]{\{ #1 \}}
\newcommand{\Set}[1]{\bigl\{ #1 \bigr\}}

\newcommand{\setdescr}[3][\mid]{\set{ #2 #1 #3 }}
\newcommand{\Setdescr}[3][|]%
     {\ifthenelse{\equal{#1}{;}}%
     {\Set{ #2 \,;\, #3 }}
     {\ifthenelse{\equal{#1}{:}}%
     {\Set{ #2 \,:\, #3 }}
     {\twincommandJN{\bigl\{}{#2\,}{\bigl#1}{\bigr}{\,#3}{\bigr\}}}}}
\newcommand{\SETDESCR}[3][|]%
     {\twincommandJN{\left\{}{#2\,}{\left#1}{\right}{\,#3}{\right\}}}

\newcommand{\Setdescrbrackets}[3][|]%
     {\twincommandJN{\bigl[}{#2}{\bigl#1}{\bigr}{\,#3}{\bigr]}}
\newcommand{\SETDESCRBRACKETS}[3][|]%
     {\twincommandJN{\left[}{#2}{\left#1}{\right}{\,#3}{\right]}}

\newcommand{\Setsize}[1]{\bigl\lvert#1\bigr\rvert}
\newcommand{\setsize}[1]{\lvert#1\rvert}

\newcommand{\intersection}{\cap}

\newcommand{\union}{\cup}

\newcommand{\olnot}[1]{\overline{#1}}

\newcommand{\nvar}{n}

\newcommand{\nclause}{m}

\newcommand{\clwidth}{k}

\newcommand{\randkcnfnclwrepl}[3][\clwidth]%
        {\ensuremath{\mathcal{F}^{#2, #3}_{#1}}}
\newcommand{\randkcnfnclwreplstd}%
        {\randkcnfnclwrepl{\clwidth}{\nvar}{\nclause}}

\newcommand{\complclassformat}[1]%
        {\textrm{\upshape{\textsf{#1}}}\xspace}
\newcommand{\cocomplclass}[1]%
        {\textrm{\upshape{\textsf{co#1}}}\xspace}

\newcommand{\DTIMEadviceclass}[2]%
    {\ensuremath{\complclassformat{DTIME}\bigl(#1\bigr)/{#2}}}

\newcommand{\NP}{\complclassformat{NP}}

\newcommand{\PCPalph}[5]%
    {\ensuremath{\complclassformat{PCP}_{{#1},{#2}}[{#3}, {#4}, {#5}]}}
\newcommand{\PCP}[4]%
    {\ensuremath{\complclassformat{PCP}_{{#1},{#2}}[{#3}, {#4}]}}

\newcommand{\introduceterm}[1]{{\emph{#1}}}

\newcommand{\eqperiod}{\enspace .}
\newcommand{\eqcomma}{\enspace ,}

\newcommand{\wrt}{with respect to\xspace}

\ifthenelse{\boolean{detectedToC}}{}
  { %
    }
\ifthenelse{\boolean{detectedToC}}{}{
\newcommand{\ie}{i.e.,\ }

}

\ifthenelse{\boolean{detectedLIPIcs} \or \boolean{detectedIJCAI}}
{\renewcommand{\st}{\errmessage{Please do not use st}}}
{\ifthenelse{\isundefined{\st}}
  {\newcommand{\st}{such that\xspace}}
  {}}      

\newcommand{\etal}{et al.\@\xspace}

\newcommand{\ifaoif}{if and only if\xspace}
\newcommand{\wolog}{without loss of generality\xspace}

\ifthenelse{\boolean{detectedIEEE}}{}{}

\newcommand{\refsec}[1]{Section~\ref{#1}}

\newcommand{\reffig}[1]{Figure~\ref{#1}}

\newcommand{\reftwofigs}[2]{Figures~\ref{#1} and~\ref{#2}}

\newcommand{\refth}[1]{Theorem~\ref{#1}}

\newcommand{\reflem}[1]{Lemma~\ref{#1}}

\newcommand{\refpr}[1]{Proposition~\ref{#1}}

\newcommand{\refcor}[1]{Corollary~\ref{#1}}

\newcommand{\refex}[1]{Example~\ref{#1}}

\ifthenelse
{\isundefined{\refeq}}
{\newcommand{\refeq}[1]{\eqref{#1}}}
{\renewcommand{\refeq}[1]{\eqref{#1}}}

\ifthenelse{\isundefined{\DONOTLOADHYPERREFPACKAGE}
  \and \not \boolean{detectedSTOC}        \and \not \boolean{detectedFOCS}
  \and \not \boolean{detectedPoster}      \and \not \boolean{detectedElsevier} 
  \and \not \boolean{detectedSIAM}        \and \not \boolean{detectedACM}
  \and \not \boolean{detectedIEEE}        \and \not \boolean{detectedNOW}
  \and \not \boolean{detectedToC}         \and \not \boolean{detectedThesis}
  \and \not \boolean{detectedLIPIcs}      \and \not \boolean{detectedSIAM}
  \and \not \boolean{detectedAAAI}        \and \not \boolean{detectedIJCAI}
  \and \not \boolean{detectedSigplanconf} \and \not \boolean{detectedACMconf}   
  \and \not \boolean{detectedCompCplx} \and \not \boolean{detectedEasyChair}}
{\usepackage[colorlinks=true,citecolor=blue]{hyperref}}
{}

\ifthenelse{\boolean{detectedSTOC}}                  
{\input{definitions-STOC.tex}}
{\ifthenelse{\boolean{detectedSIAM}}
  {\input{definitions-SIAM.tex}}
  {\ifthenelse{\boolean{detectedNOW}}
    {\input{definitions-NOW.tex}}
    {\ifthenelse{\boolean{detectedLMCS}}
      {\input{definitions-lmcs.tex}}
      {\ifthenelse{\boolean{detectedIEEE}}
        {\input{definitions-IEEE.tex}}
        {\ifthenelse{\boolean{detectedLNCS}}
          {\input{definitions-LNCS.tex}}
          {\ifthenelse{\boolean{detectedACM} \or \boolean{detectedACMconf}}
            {\input{definitions-ACM.tex}}
            {\ifthenelse{\boolean{detectedToC}}
              {\input{definitions-ToC.tex}}
              {\ifthenelse{\boolean{detectedLIPIcs}}
                {\input{definitions-LIPIcs.tex}}
                {\ifthenelse{\boolean{detectedCompCplx}}
                  {\input{definitions-CompCplx.tex}}
                  {\ifthenelse{\boolean{detectedSigplanconf}}
                    {\input{definitions-Sigplanconf.tex}}
                    {\ifthenelse{\boolean{detectedAAAI}}
                      {}
                      {\ifthenelse{\boolean{detectedArticle} 
                          \or \boolean{detectedElsevier}
                          \or \boolean{detectedEasyChair}}
                        {%
\newtheorem{standardlocalcounter}{Dummy}[section]
\newtheorem{standardglobalcounter}{Dummy}

\theoremstyle{plain}    

\newtheorem{theorem}[standardlocalcounter]{Theorem}
\newtheorem{lemma}[standardlocalcounter]{Lemma}
\newtheorem{proposition}[standardlocalcounter]{Proposition}
\newtheorem{corollary}[standardlocalcounter]{Corollary}

\newtheorem{openproblem}[standardglobalcounter]{Open Problem}

\theoremstyle{definition}

\newtheorem{claim}[standardlocalcounter]{Claim}

\theoremstyle{remark}

\newtheorem{example}[standardlocalcounter]{Example}

}
                        {\input{definitions-report.tex}}}}}}}}}}}}}}

\ifthenelse{\boolean{detectedThesis}}
{}
{\usepackage{ifpdf}}

\ifthenelse{\boolean{detectedToC}}
{}
{\usepackage{graphicx}}

\ifpdf         
\fi

\ifthenelse
{\boolean{detectedSTOC}   \or \boolean{detectedThesis} \or 
 \boolean{detectedLNCS}   \or \boolean{detectedToC}    \or 
 \boolean{detectedLIPIcs} \or \boolean{detectedAAAI}   \or
 \boolean{detectedIJCAI}  \or \boolean{detectedSIAM}}
{}
{
  \setcounter{secnumdepth}{3}
  \setcounter{tocdepth}{3}
}

\newcount\hours
\newcount\minutes
\def\SetTime{\hours=\time
\global\divide\hours by 60
\minutes=\hours
\multiply\minutes by 60
\advance\minutes by-\time
\global\multiply\minutes by-1 }
\SetTime
\def\now{\number\hours:\ifnum\minutes<10 0\fi\number\minutes}

\newcommand{\proofstd}{\pi}

\newcommand{\varx}{\ensuremath{x}}
\ifthenelse{\boolean{detectedACMconf}}
{\providecommand{\vary}{\ensuremath{y}}}
{\newcommand{\vary}{\ensuremath{y}}}

\ifthenelse{\boolean{detectedACMconf}}
{
 }
{
 }

\newcommand{\SETSOFVARSORLIT}[2]%
        {\mathit{#1}\left({#2}\right)}
\newcommand{\setsofvarsorlit}[2]%
        {\mathit{#1}({#2})}
\newcommand{\Setsofvarsorlit}[2]%
        {\mathit{#1}\bigl({#2}\bigr)}

\newcommand{\restrict}[2]{{{#1}\!\!\upharpoonright_{#2}}}

\newcommand{\derivabbrev}[2]{\bigl( #1 \vdash #2 \bigr)}
\newcommand{\derivabbrevsmall}[2]{( #1 \vdash #2 )}
\newcommand{\derivabbrevcompact}[2]{\bigl( #1 \vdash #2 \bigr)}

\newcommand{\refutabbrevsmall}[1]{\derivabbrevsmall{#1}{\!\bot}}
\newcommand{\refutabbrevcompact}[1]{\derivabbrevcompact{#1}{\!\bot}}

\newcommand{\genericrefsmall}[3]%
    {{\mathit{#1}}_{#2}\refutabbrevsmall{#3}}
\newcommand{\genericrefcompact}[3]%
    {{\mathit{#1}}_{#2}\refutabbrevcompact{#3}}
\newcommand{\genericderiv}[4]%
    {{\mathit{#1}}_{#2}\derivabbrev{#3}{#4}}
\newcommand{\genericderivsmall}[4]%
    {{\mathit{#1}}_{#2}\derivabbrevsmall{#3}{#4}}
\newcommand{\genericderivcompact}[4]%
    {{\mathit{#1}}_{#2}\derivabbrevcompact{#3}{#4}}
\newcommand{\generictaut}[3]%
    {{\mathit{#1}}_{#2}\derivabbrev{}{#3}}
\newcommand{\generictautcompact}[3]%
    {{\mathit{#1}}_{#2}\derivabbrevcompact{}{#3}}
\newcommand{\generictautsmall}[3]%
    {{\mathit{#1}}_{#2}\derivabbrevsmall{}{#3}}

\newcommand{\formulaformat}[1]{\mathit{#1}}

\newcommand{\extendedversion}[1]{\widetilde{#1}}

\newcommand{\epopnot}[1]%
    {\extendedversion{\formulaformat{POP}}_{#1}}

\newcommand{\elopnot}[1]%
    {\extendedversion{\formulaformat{LOP}}_{#1}}

\newcommand{\ephpnot}[2]%
    {\vphantom{\extendedversion{\formulaformat{PHP}}}
      {\smash{\extendedversion{\formulaformat{PHP}}}
        \vphantom{\formulaformat{PHP}}}^{#1}_{#2}}

\newcommand{\efphpnot}[2]%
    {\vphantom{\extendedversion{\formulaformat{FPHP}}}
      {\smash{\extendedversion{\formulaformat{FPHP}}}
        \vphantom{\formulaformat{FPHP}}}^{#1}_{#2}}

\newcommand{\ontophpnot}[2]%
    {\formulaformat{Onto}\text{-}\formulaformat{PHP}^{#1}_{#2}}
\newcommand{\ontofphpnot}[2]%
    {\formulaformat{Onto}\text{-}\formulaformat{FPHP}^{#1}_{#2}}

\newcommand{\graphontophpnot}[1][G]%
    {\text{$\formulaformat{Onto}$-$\formulaformat{PHP}$}({#1})}

\newcommand{\perfectmatchingnot}[1][G]%
    {\formulaformat{PM}({#1})}

\newcommand{\precolgadgetsize}{M}
\newcommand{\cplengthconst}{\kappa}
\newcommand{\cplength}{L}
\newcommand{\cpcol}{c}

\newcommand{\rcvertex}[1]{r_{#1}}
\newcommand{\lcvertex}[1]{\ell_{#1}}
\newcommand{\clrvertex}[1]{\gamma_{#1}}
\newcommand{\clrvindex}{t}

\newcommand{\substarrowdisplayed}{\ \mapsto \ }

\newcommand{\permsigma}{\sigma}

\newcommand{\vcolouring}{\chi}
\newcommand{\vcolour}[1]{\vcolouring({#1})}
\newcommand{\flightviaedge}[2]{{#1} \leftarrow {#2}}
\newcommand{\flightnotviaedge}[2]{{#1} \not\leftarrow {#2}}

\newcommand{\alldiffcolgadgets}{\widehat{\graphg}}

\newcommand{\diffcolgadget}[4]%
        {\graphg_{({#1},{#2}) \not\leftarrow ({#3},{#4})}}
\newcommand{\diffcolgadgetstd}%
        {\diffcolgadget{\pigeonindex}{\pigeonindex'}
        {\colstd}{\colstd'}}
\newcommand{\diffcolgadgetsamestd}%
        {\diffcolgadget{\pigeonindex}{\pigeonindex'}
        {\colstd}{\colstd}}

\newcommand{\pigeonset}{I}
\newcommand{\pigeonindex}{i}
\newcommand{\holeset}{J}
\newcommand{\holeindex}{j}
\newcommand{\holesforpigeon}[1]{J({#1})}

\newcommand{\xcolouring}[1]{${#1}$\nobreakdash-colouring\xspace}

\newcommand{\xcolourable}[1]{${#1}$\nobreakdash-colourable\xspace}
\newcommand{\xcolourability}[1]{${#1}$\nobreakdash-colourability\xspace}

\newcommand{\kcolouring}{$\numcolours$\nobreakdash-colouring\xspace}
\newcommand{\kcoloured}{$\numcolours$\nobreakdash-coloured\xspace}
\newcommand{\kcolourable}{$\numcolours$\nobreakdash-colourable\xspace}
\newcommand{\kcolourability}{$\numcolours$\nobreakdash-colourability\xspace}

\newcommand{\E}{\mathbb{E}}

\newcommand{\colstd}{c}
\newcommand{\colalt}{b}
\newcommand{\numcolours}{k}
\newcommand{\numcolors}{\numcolours}
\newcommand{\graphg}{G}
\newcommand{\graphstd}{\graphg}
\newcommand{\bipartstd}{B}
\newcommand{\colouringstd}{\chi}

\newcommand{\chomatic}{\chi}
\newcommand{\runity}{w}
\renewcommand{\runity}{\omega}

\newcommand{\PC}{PC\xspace}
\newcommand{\PCR}{PCR\xspace}
\newcommand{\degreestd}{d}

\newcommand{\poly}[1]{\uppercase{#1}}

\newcommand{\polyp}{\poly{p}}
\newcommand{\polyq}{\poly{q}}
\newcommand{\polyr}{\poly{r}}

\newcommand{\polyset}[1]{\mathcal{\uppercase{#1}}}

\newcommand{\polysets}{\polyset{S}}

\newcommand{\dualvar}[1]{\olnot{#1}}
\newcommand{\dvarx}{\dualvar{x}}

\providecommand{\stoptime}{\tau}
\renewcommand{\stoptime}{\tau}
\newcommand{\timet}{t}

\newcommand{\CP}{CP\xspace}
\newcommand{\lin}[1][]{A_{#1}}
\newcommand{\linaux}[1][]{B_{#1}}
\newcommand{\lincoeff}[1][]{a_{#1}}
\newcommand{\lindiv}{c}
\newcommand{\linconst}{\gamma}
\newcommand{\linset}[1]{\mathcal{\uppercase{#1}}}

\newcommand{\partassign}{\rho}
\newcommand{\partundefined}{*}
\newcommand{\partdomain}{D}
\newcommand{\partelem}{d}
\newcommand{\partrange}{R}

\DeclareMathOperator{\dom}{dom}
\renewcommand{\domainof}[1]{\ensuremath{\dom ( #1 )}}

\newcommand{\theauthorML}{the first author\xspace}

\newcommand{\TheauthorJN}{The second author\xspace}

\ifthenelse{\boolean{conferenceversion}}
{

}
{

}

\ifthenelse{\boolean{conferenceversion}}
{}
{
\newtheoremstyle{metacommenttheoremstyle}%
    {3pt}%
    {3pt}%
    {\sffamily \itshape \scriptsize
    }%
    {}%
    {\bfseries \scshape \footnotesize }%
    {:}%
    { }%
    {}%

\theoremstyle{metacommenttheoremstyle}
}
\newtheorem{jncommentcontainer}{Jakob's comment}
\newtheorem{mlcommentcontainer}{Massimo's comment}
\newtheorem{mmcommentcontainer}{Mladen's comment}
\newtheorem{mvcommentcontainer}{Marc's comment}

\ifthenelse{\isundefined{\DONOTINSERTCOMMENTS}}
{ %
  \usepackage[usenames,dvipsnames]{xcolor}

  \newcommand{\jncomment}[1]%
  {\begin{jncommentcontainer} \textcolor{blue}{#1} \end{jncommentcontainer}}

  \newcommand{\mlcomment}[1]%
  {\begin{mlcommentcontainer} \textcolor{OliveGreen}{#1} \end{mlcommentcontainer}}

  \newcommand{\mmcomment}[1]%
  {\begin{mmcommentcontainer} \textcolor{magenta}{#1} \end{mmcommentcontainer}}

  \newcommand{\mvcomment}[1]%
  {\begin{mvcommentcontainer} \textcolor{orange}{#1} \end{mvcommentcontainer}}

}
{ %
  \newcommand{\jncomment}[1]{}
  \newcommand{\mlcomment}[1]{}
  \newcommand{\mmcomment}[1]{}
  \newcommand{\mvcomment}[1]{}
}

\ifthenelse{\boolean{conferenceversion}}
{}
{
  \numberwithin{equation}{section}
}

\begin{document}

\title{%
  Graph Colouring is Hard for Algorithms Based on \\
  Hilbert's Nullstellensatz and Gröbner Bases%
  \thanks{This is the full-length version of the paper with the same
    title that appeared in
    \emph{Proceedings of the 32nd Annual
      Computational Complexity Conference ({CCC}~'17)}.}}

\author{%
  Massimo Lauria \\
  Sapienza --- Università di Roma, Italy\\
  \texttt{massimo.lauria@uniroma1.it}
  \and
  Jakob Nordström \\
  University of Copenhagen, Denmark, and Lund University, Sweden\\
  \texttt{jn@di.ku.dk}
}

\date{\today}

\maketitle

\ifthenelse{\boolean{conferenceversion}}
{}
{
  \thispagestyle{empty}

  \pagestyle{fancy}
  \fancyhead{}
  \fancyfoot{}
  \renewcommand{\headrulewidth}{0pt}
  \renewcommand{\footrulewidth}{0pt}

  \fancyhead[CE]{\slshape 
    GRAPH COLOURING IS HARD FOR
    HILBERT'S NULLSTELLENSATZ AND GRÖBNER BASES}
  \fancyhead[CO]{\slshape \nouppercase{\leftmark}}
  \fancyfoot[C]{\thepage}

  \setlength{\headheight}{13.6pt}
}
\begin{abstract}
  We consider the graph \mbox{$k$-colouring problem} encoded 
  as a set of polynomial equations
  in the standard way
  over \mbox{$0/1$-valued} variables.
  We prove that there are bounded-degree graphs that do not have legal
  $k$-colourings but for which the polynomial calculus proof system
  defined in [Clegg \etal~'96, Alekhnovich \etal~'02]
  requires linear degree, and hence exponential size, to establish
  this fact. This implies a linear degree lower bound for any algorithms based
  on Gröbner bases solving 
  graph \mbox{$k$-colouring}
  using this encoding.
  The same bound applies also for the algorithm studied in a sequence of papers
  [De Loera \etal~'08,~'09,~'11,~'15]
  based on Hilbert's Nullstellensatz proofs for a slightly
  different encoding, thus resolving an open problem  mentioned
  in [De Loera \etal~'08,~'09,~'11]  and  [Li \etal~'16].
  We obtain our results by combining the polynomial calculus degree lower
  bound for functional pigeonhole principle (FPHP) formulas over
  bounded-degree bipartite graphs in
  [Mik\v{s}a and Nordström~'15]
  with a reduction from FPHP to $k$-colouring
  derivable by polynomial calculus in constant degree.
\end{abstract}

\jncomment{Check that the paper reflects fairly that we \\
  (a) get very strong lower bounds for the De Loera \etal standard
  algorithm, but that \\ 
  (b) our lower bounds don't obviously extend to their enhanced
  version dealing with symmetries, added redudant polynomials,
  alternative versions of the Nullstellensatz, et cetera.} 
\section{Introduction}
\label{sec:intro}

Given an undirected graph
$G = (V,E)$ and a positive integer~$\numcolours$,  
can the vertices
$v \in V$ be coloured with at most $\numcolours$~colours so that no
vertices connected by an edge
have the same colour?
This
\introduceterm{graph colouring problem}
is perhaps one of the most extensively studied
\NP-complete problems.
It is widely believed that any algorithm for this problem has to run
in exponential time in the worst case, and indeed the currently
fastest algorithm for $3$\nobreakdash-colouring runs in time $O(1.3289^{n})$
\cite{Beigel05Coloring}.
A survey on various algorithms and techniques for so-called exact
algorithms is~\cite{Husfeldt15Colouring}.

Many
graph colouring
instances
of interest might not exhibit
worst-case behaviour, 
however,
and therefore it makes sense to study algorithms
without worst-case guarantees and examine how they perform in
practice. Dually, it can be of interest to study weak models of
computation, which are nevertheless strong enough to capture the power
of such algorithms, and prove unconditional lower bounds 
for these
models. 
Obtaining such lower bounds is
the goal of this work.

\subsection{Brief Background}

Since current state-of-the-art algorithms for
propositional satisfiability such as
\introduceterm{conflict-driven clause learning (CDCL)}~%
\cite{BS97UsingCSP,MS99Grasp,MMZZM01Engineering}
are ultimately based on 
\ifthenelse{\boolean{conferenceversion}}
{\introduceterm{resolution}~\cite{Blake37Thesis},} 
{the \introduceterm{resolution proof system}~\cite{Blake37Thesis},} 
it is perhaps not so surprising that this approach can be used to solve
\ifthenelse{\boolean{conferenceversion}}
{colouring} 
{graph colouring} 
problems as well.
According to~\cite{BCCM05RandomGraph}, McDiarmid developed a method
for deciding \kcolourability that captures many concrete
algorithms~\cite{McDiarmid84Colouring}. This method, 
viewed
as a proof
system, is simulated by resolution.

There are exponential lower bounds for resolution proofs of 
non-\kcolourability that apply to any such method. In particular, 
the paper
\cite{BCCM05RandomGraph}~presents
average-case exponential lower bounds for random graph \kcolouring
instances sampled so that the graphs are highly likely not to be
\kcolourable.
This result is obtained by
proving width lower bounds, \ie
lower bounds on the size of a largest clause in any 
resolution refutations of the formulas, 
and then using that linear width lower
bounds implies exponential size lower bounds~\cite{BW01ShortProofs}.

Another possible approach is to attack the \kcolouring problem using
algebra. Various algebraic methods have been considered in
\ifthenelse{\boolean{conferenceversion}}
{\cite{AT92Colorings,Lovasz94Stable,Matiyasevich74Criterion,Matiyasevich04Algebraic}.}  
{\cite{Matiyasevich74Criterion,Matiyasevich04Algebraic,Lovasz94Stable,AT92Colorings}.}  
The thesis \cite{Bayer82Division} contains the first explicit attempt we know of
to encode the \xcolouring{3} problem using Hilbert's Nullstellensatz.
At a high level, the idea is to 
write the problem as a set of polynomial equations
$\setdescr{f_i(x_1,\ldots, x_n) = 0}{i \in [m]}$
over a suitable field~$\F$
so that legal colourings correspond to solutions,
and if this is done in the right way it holds that this
system of equations has no solution if and only if there are
polynomials
$g_1, \ldots, g_m$ 
such that
$
\sum_{i=1}^{m} g_i f_i = 1
$.
This latter equality is referred to as a 
\introduceterm{Nullstellensatz certificate}
of non-colourability, and the \introduceterm{degree} of this
certificate is the largest degree of any polynomial  $g_i f_i$ in the sum.
Later papers based on Nullstellensatz and Gröbner bases  such as
\ifthenelse{\boolean{conferenceversion}}
{\cite{DeLoera95Grobner,HW08Algebraic,Mnuk01Representing}}
{\cite{DeLoera95Grobner,Mnuk01Representing,HW08Algebraic}}
have attracted a fair amount of attention.
For this work, 
we are particularly interested in the sequence of papers
\cite{DLMM08Hilbert,DLMO09ExpressingCombinatorial,DLMM11ComputingInfeasibility,DMPRRSSS15GraphColouring},
which uses an encoding of the \kcolouring problem that will be
discussed more in detail later in the paper.

There seem to be no formally proven lower bounds for these algebraic
methods.  On the contrary, the authors
of~\cite{DLMO09ExpressingCombinatorial} report that essentially all of
the benchmarks they have studied have Nullstellensatz certificates of
constant (and very small) degree.
Indeed, no lower bounds for graph colouring is known for the 
corresponding proof systems
\introduceterm{Nullstellensatz}~\cite{BIKPP94LowerBounds}
or the stronger system
\introduceterm{polynomial calculus}
\ifthenelse{\boolean{conferenceversion}}
{\cite{ABRW02SpaceComplexity,CEI96Groebner}.}
{\cite{CEI96Groebner,ABRW02SpaceComplexity}.}
Intriguingly, in a close parallel to the case for resolution it is
known that strong enough lower bounds on polynomial calculus degree
imply exponential lower bounds on proof size~\cite{IPS99LowerBounds},
but the techniques for proving degree lower bounds are much less
developed than the width lower bound techniques for resolution.

There have been degree lower bounds proven for polynomial calculus
refutations of concrete systems of polynomial equations, but in most
cases these have been encodings of obviously false statements
(such as negations of the pigeonhole principle or graph handshaking
lemma), rather than computationally hard problems.
For some of these problems
degree lower bounds can be obtained by making an affine transformation from
$\set{0,1}$\nobreakdash-valued variables to
$\set{-1,+1}$\nobreakdash-valued variables
\ifthenelse{\boolean{conferenceversion}}
{\cite{BI10Random,BGIP01LinearGaps},}
{\cite{BGIP01LinearGaps,BI10Random},}
but this only works for polynomial equations with the right structure
and only for fields of characteristic distinct from~$2$.
A  general and powerful method, which is independent of the field
characteristic, was developed in~\cite{AR03LowerBounds}, but has
turned out to be not so easy to 
to apply (except in a few papers such as
\cite{GL10Automatizability,GL10Optimality}).
A slightly different, 
and formally speaking somewhat incomparable, 
version of the approach in~\cite{AR03LowerBounds} was recently
presented in~\cite{MN15GeneralizedMethodDegree}, and this latter paper
also more clearly highlighted the similarities and differences between
resolution width lower bound techniques and polynomial calculus degree
lower bound techniques. 
The new framework in~\cite{MN15GeneralizedMethodDegree} was used to
establish a new degree lower bound which plays a key role in our
paper.

\subsection{Our Contributions}

\ifthenelse{\boolean{conferenceversion}}
{We exhibit families}
{We exhibit explicit families}
of non-\kcolourable graphs of 
\ifthenelse{\boolean{conferenceversion}}
{bounded  degree} 
{bounded vertex degree} 
such that the canonical encoding of the corresponding
\kcolouring instances into systems of polynomial equations over
$\set{0,1}$\nobreakdash-valued variables require linear degree to be
refuted in polynomial calculus.

\begin{theorem}[informal]
  \label{th:main-theorem}
  For any constant $\numcolours \geq 3$ there are explicit families of
  graphs $\{\graphstd_{n}\}_{n\in \N}$ of size~$\bigoh{n}$ and
  constant vertex degree, which are not \kcolourable but for which the
  polynomial calculus proof system requires linear degree, and hence
  exponential size, to prove this fact, regardless of the underlying
  field. 
\end{theorem}

Our degree lower bound also applies to a slightly different encoding
with primitive $\numcolours$th roots of unity used
in~\cite{DLMM08Hilbert,DLMM11ComputingInfeasibility} 
to build
\kcolouring algorithms based on Hilbert's Nullstellensatz. These
algorithms construct certificates of non-$\numcolours$-colourability by 
solving linear systems of equations over the coefficients of all
monomials up to a certain degree. 

Just as the algorithms
in~\cite{DLMM08Hilbert,DLMM11ComputingInfeasibility}, our lower bound
does not work for all fields (the field must have an extension
field in which there is a primitive $\numcolours$th root of unity).
For simplicity, we state below a concrete result for Nullstellensatz
certificates over~$\gf{2}$ for non-\xcolourability{3}, which is
one of the main cases considered
in~\cite{DLMM08Hilbert,DLMM11ComputingInfeasibility}.

\begin{corollary}
  \label{cor:main-cor}
  There are explicit families of non-\xcolourable{3} graphs such that
  the algorithms based on Hilbert's Nullstellensatz
  over~$\gf{2}$
  in~\cite{DLMM08Hilbert,DLMM11ComputingInfeasibility} need to find
  certificates of linear degree, and hence must solve systems of
  linear equations of exponential size, in order to certify
  non-\xcolourability{3}.
\end{corollary}

We remark that \refcor{cor:main-cor} answers an open question raised
in, for example,
\cite{DLMM08Hilbert,DLMO09ExpressingCombinatorial,DLMM11ComputingInfeasibility,LLO16LowDegreeColorability}.

Finally, we want to mention that the graph colouring instances that we
construct turn out to be easy for the proof system
\introduceterm{cutting  planes}~\cite{CCT87ComplexityCP},
which formalizes the integer linear programming algorithm 
\ifthenelse{\boolean{conferenceversion}}
{in~\cite{Chvatal73EdmondPolytopes,Gomory63AlgorithmIntegerSolutions}}
{in~\cite{Gomory63AlgorithmIntegerSolutions,Chvatal73EdmondPolytopes}}
and underlies so-called
\introduceterm{pseudo-Boolean} SAT solvers
such as, for instance, 
\introduceterm{Sat4j}~\cite{LP10Sat4j,Sat4j}.

\begin{proposition}
  \label{pr:cp-refutation}
  The graph colouring instances for the non-\kcolourable graphs in
  \refth{th:main-theorem} have polynomial-size refutations in the cutting
  planes proof system.
\end{proposition}

\subsection{Techniques}

Perhaps somewhat surprisingly, no heavy-duty machinery is required to
establish \refth{th:main-theorem}. Instead, all that is needed is a
nifty
reduction.
Our starting point is the so-called 
\introduceterm{functional pigeonhole principle (FPHP) formula}  
restricted to a bipartite graph of bounded left degree~$\numcolours$.
This formula expresses the claim that a set of pigeons 
$\pigeonindex \in \pigeonset$
can be mapped to a set of pigeonholes 
$\holeindex \in \holeset$
in a one-to-one fashion,
where in addition the pigeons are constrained so that every pigeon can
choose not between all available holes but only between a set of
$\numcolours$~holes 
as specified by the bipartite graph.
Clearly, FPHP formulas are unsatisfiable 
when  $\setsize{\pigeonset} > \setsize{\holeset}$.

Any instance of a graph FPHP formula
can be viewed as a constraint satisfaction
problem by ordering the available holes for every pigeon in some
arbitrary but fixed way, and then keeping track of where each pigeon
is mapped by recording the ordinal number of its chosen pigeonhole.
If the
$\colstd$th hole for pigeon~$\pigeonindex$ and the
$\colstd'$th hole for pigeon~$\pigeonindex'$ is one and the same
hole~$\holeindex$, 
then
pigeons~$\pigeonindex$ and~$\pigeonindex'$ cannot be allowed to make
choices
$\colstd$ and~$\colstd'$ simultaneously.
If we view this constraint as an edge in graph with the
pigeons~$\pigeonset$ as 
vertices, this is already close to a graph colouring instance, except
that what is forbidden for the neighbours
$\pigeonindex$ and~$\pigeonindex'$
is not the same colour~$\colstd$
but some arbitrary pair of possibly distinct colours
$(\colstd, \colstd')$.
However, the idea outlined above can be turned into a proper reduction
from graph FPHP formulas to \kcolouring instances by using
appropriately constructed gadgets of constant size.

We then combine this reduction with the recent polynomial calculus degree
lower bound in~\cite{MN15GeneralizedMethodDegree},
which works as long as the underlying bipartite graph is a 
\introduceterm{boundary expander}
(a.k.a.~\introduceterm{unique-neighbour expander}).
More precisely, we show that the reduction from FPHP to graph
\kcolouring sketched above can be computed in polynomial calculus in
low degree.  Therefore, any low-degree polynomial calculus refutations
of the graph \kcolouring instances could be used to obtain low-degree
refutations of FPHP instances, but \cite{MN15GeneralizedMethodDegree}
tells us that FPHP instances over expander graphs require linear degree.

In order to obtain \refcor{cor:main-cor},
we assume that we have a low-degree Nullstellensatz certificate 
(or, more generally, a polynomial calculus proof) of
non-colourability for the roots-of-unity encoding 
in~\cite{DLMM08Hilbert,DLMM11ComputingInfeasibility}.
Then it is not hard to show that if the field we are working in
contains a primitive $\numcolours$th root of unity,  
we can apply a linear variable substitution to obtain a
polynomial calculus refutation in essentially the same degree of the
colouring instance in the encoding with 
$\set{0,1}$\nobreakdash-valued variables.
The corollary now follows from \refth{th:main-theorem}. 

As should be clear from the discussion above, the hardness of our
graph colouring instances ultimately derives from the pigeonhole
principle. However, this combinatorial principle is well-known to be
easy for cutting planes. 
We establish
\refpr{pr:cp-refutation}
by showing that cutting planes can unpack the reduction described
above to recover the original pigeonhole principle instance, after which this
instance can be efficiently refuted.

\subsection{Outline of This Paper}

The rest of this paper is organized as follows.  We start by
presenting some proof complexity preliminaries and discussing how to
encode the graph colouring problem in \refsec{sec:preliminaries}.
In
\refsec{sec:pc-lower-bound}
we describe our graph \kcolouring instances and prove that they are
hard for polynomial calculus, and in
\refsec{sec:cp-refutation}
we show that the same instances are easy for cutting planes.
We conclude in 
\refsec{sec:conclusion}
by discussing some directions for future research.
\ifthenelse{\boolean{conferenceversion}}
{We refer to the upcoming full-length version for all missing proofs.}
{}

\section{Preliminaries}
\label{sec:preliminaries}

Throughout this paper $\varx_{1}, \ldots,\varx_{n}$ denote
\mbox{$\set{0,1}$-valued} variables, where we think of $1$~as true and
$0$~as false.  We write $\N = \set{0,1,2,\ldots}$ for the natural
numbers and denote $\Nplus = \N \setminus \set{0}$.  For
$n \in \Nplus$ we use the standard notation
$[n] = \set{1,2,\ldots,n}$.  For a set~$E$, we use the shorthand
$e \neq e' \in E$ to index over pairs of distinct elements
$e, e' \in E$, $e \neq e'$.

\subsection{Proof Complexity}

\introduceterm{Polynomial calculus (PC)}~\cite{CEI96Groebner} 
is a proof system based on algebraic reasoning where one expresses constraints
over Boolean variables as polynomial equations and applies algebraic
manipulations to deduce new equations.
The constraints are over \mbox{$\set{0,1}$-valued} variables
$\varx_{1}, \ldots,\varx_{n}$, and each constraint 
is encoded as a polynomial $\polyq$ in
the ring $\F[\varx_{1},\ldots,\varx_{n}]$, where $\F$ is some fixed field.
The intended meaning is that $\polyq=0$ if and only if the constraint
is satisfied, but we omit ``$=0$'' below and only write the
polynomial~$\polyq$.  
A \introduceterm{\PC derivation} of a polynomial~$\polyr$ from a set of
polynomials 
$\polysets = \{\polyq_1,\ldots,\polyq_m\}$ 
is a sequence 
$(\polyp_{1}, \ldots, \polyp_\stoptime)$ 
such that $\polyp_{\stoptime}=\polyr$ and for $1 \leq \timet \leq
\stoptime$ the polynomial 
$\polyp_{\timet}$ is obtained by one of the following derivation rules:
\begin{itemize}       \itemsep=0pt
\item \textbf{Boolean axiom:} 
  $\polyp_{\timet}$ is $\varx^{2}-\varx$ 
  for some variable $\varx$;
\item \textbf{Initial axiom:} $\polyp_{\timet}$ is one of the
  polynomials $ \polyq_{j} \in \polysets$; 
\item 
  \textbf{Linear combination:} 
  $\polyp_\timet = \alpha \polyp_{i} + \beta \polyp_{j} $ for
  $ 1 \leq i,j < \timet $ and some $\alpha, \beta \in \mathbb{F}$;
\item \textbf{Multiplication:} $ \polyp_{\timet} = \varx \polyp_{i}$ 
  for $ 1 \leq i <\timet $ and some  variable $\varx$.
\end{itemize}
A
\introduceterm{PC refutation}
of~$\polysets$
is a derivation of the multiplicative identity~$1$ of~$\F$ from~$\polysets$.
Note that the Boolean axioms make sure that variables can only take
values $0$ and~$1$. For this reason, we can assume without loss of
generality that all polynomials appearing in PC derivations are
multilinear. 

The 
\introduceterm{size} of a polynomial~$\polyp$
is the number of distinct monomials in it when it is expanded out as a
linear combination of monomials,%
\footnote{Just to make terminology precise, in this paper
  a \introduceterm{monomial} is a product of variables,
  a \introduceterm{term} is a monomial multiplied by a non-zero
  coefficient from the field~$\F$, and a \introduceterm{polynomial} is
  always considered as a linear combinations of terms over
  pairwise distinct monomials.}
and the 
\introduceterm{degree} of~$\polyp$ is the largest (total) degree of
any monomial in~$\polyp$.
The size of a PC derivation~$\proofstd$ is the sum of the sizes of
all polynomials in~$\proofstd$,
and the degree is the maximal degree of any 
polynomial in~$\proofstd$.
One can also define the 
\introduceterm{length} of a PC derivation as the number of derivation
steps in it, but this not so interesting a measure since it may fail
to take account of polynomials of exponential size.%
\footnote{Indeed, if multiplication is defined to multilinearize
  polynomials automatically, as in, e.g.,~\cite{AR03LowerBounds}, then
  any unsatisfiable CNF formula encoded into polynomials in the
  natural way can be refuted in linear
  length---see~\cite{MN15GeneralizedMethodDegree} for details.}
A fundamental fact about PC is that the size and degree measures are
tightly related as stated next.

\begin{theorem}[\cite{IPS99LowerBounds}]
  \label{th:ips}
  For any set $\polysets$ of inconsistent polynomials of degree at
  most $\degreestd'$ over $n$ variables
  it holds that 
  if the minimum degree of any \PC
  refutation for $\polyset{S}$ is at least $\degreestd$, then 
  any \PC refutation of~$\polyset{S}$
  has size
  $\exp \bigl( \Bigomega{(\degreestd-\degreestd')^{2} / n } \bigr)$.
\end{theorem}
In particular,
if the polynomials in 
$\polysets$
have constant degree
but require refutations of 
degree linear in the number of variables~$n$, then
any refutation must have exponential size.

We remark that there is also a slightly more general version of this
proof system known as
\introduceterm{polynomial calculus (with) resolution
  (\PCR)}~\cite{ABRW02SpaceComplexity}.
The difference is that PCR has separate formal
variables~$\varx$ and~$\dvarx$ to represent both positive and negative
literals when translating CNF formulas into sets of polynomials,  as
well as \introduceterm{complementarity axioms}
$\varx + \dvarx -  1$ 
to ensure that  
$\varx$ and~$\dvarx$ 
take opposite values. This yields a nicer and
more well-behaved proof system. The change from PC to PCR does not
affect the degree needed to refute an inconsistent set of polynomial
equations, however, and
\refth{th:ips}
holds also for PCR. Therefore, the lower bounds we show in this paper
apply both to PC and PCR.
The presence of Boolean axioms allows to derive
$\prod_{i}x_{i}^{\ell_{i}} - \prod_{i}x_{i}$ for $\ell_{i} \geq 1$ in
degree $\sum_{i}\ell_{i}$ and polynomial size. Therefore whenever we
need to derive some polynomial it is sufficient to derive its
multilinear version.

Another aspect worth noticing is that it makes perfect sense to define
polynomial calculus also for sets of polynomial equations that do not
include Boolean axioms
$\varx^{2}-\varx$. 
One variant studied in the literature is to 
add
include 
axioms
$\varx^{\numcolours} - 1$ instead, \ie to insist that
the value of~$\varx$ should be a \mbox{$\numcolours$th root} of unity.
In such a setting it is no longer necessarily true that large degree
implies large size, however.

In this paper we will also consider
\introduceterm{cutting planes (\CP)}~\cite{CCT87ComplexityCP},
which is a proof system based on manipulation of  inequalities
$
  \sum_{i} \lincoeff[i]\varx_{i}  \geq \linconst
$
where $\lincoeff[i]$ and $\linconst$ are integers
and
$\varx_{1}, \ldots,\varx_{n}$
are
$\set{0,1}$-valued variables.
A \introduceterm{\CP derivation} of an inequality $\linaux$ from a set
of
inequalities~%
$\linset{S}=\{\lin[1],\ldots,\lin[m]\}$ 
is a sequence 
$(\linaux[1], \ldots, \linaux[\stoptime])$ 
such that $\linaux[\stoptime]=\linaux$ and for $1 \leq \timet \leq
\stoptime$ the inequality
$\linaux[\timet]$ is obtained by one of the
following 
derivation
rules:
\begin{itemize} \itemsep=0pt
\item \textbf{Variable axiom:} 
  $\linaux[\timet]$ is either 
  $\varx \geq 0$ or 
  $-\varx \geq -1$ 
  for some variable $\varx$.
\item \textbf{Initial axiom:}  %
  $\linaux[\timet]$ is some $\lin[j] \in \linset{S}$;
\item 
  \textbf{Sum:} 
  $\linaux[\timet] = \linaux[i] + \linaux[j] $ for
  $ 1 \leq i,j < \timet $.

\item 
  \textbf{Scalar multiplication:} 
  $\linaux[\timet] = \lindiv \linaux[i] $ for
  $ 1 \leq i < \timet $ and $\lindiv \in \mathbb{N}$;

\item
  \textbf{Division:}
  The inequality $\linaux[\timet]$ is
  \begin{equation}
    \sum_{i} \frac{\lincoeff[i]}{\lindiv}\varx_{i}  \geq 
    \CEILING{  \frac{\linconst}{\lindiv} }
\end{equation}
where $\lindiv$ divides all $\lincoeff[1], \ldots, \lincoeff[n]$ and
$\sum_{i} \lincoeff[i]\varx_{i} \geq \linconst$ 
is some inequality
$\linaux[i]$ for $1 \leq i<\timet$.
\end{itemize}
A \CP refutation of $\linset{S}=\{\lin[1],\ldots,\lin[m]\}$ is
a derivation from $\linset{S}$ of the inequality
$ 0 \geq 1$.
In what follows, we will 
often write
$  \sum_{i} \lincoeff[i]\varx_{i}  \leq \linconst $ 
as an alias for
\mbox{$  \sum_{i} - \lincoeff[i]\varx_{i}  \geq -\linconst $,}
and we will also use
$  \sum_{i} \lincoeff[i]\varx_{i}  = \linconst $ 
as a shorthand for the two inequalities
$  \sum_{i} \lincoeff[i]\varx_{i}  \leq \linconst $ 
and
$  \sum_{i} \lincoeff[i]\varx_{i}  \geq \linconst $.

The 
\introduceterm{length} of a CP derivation is the number of derivation
steps in it.
The 
\introduceterm{size} of a linear inequality
$\sum_{i} \lincoeff[i]\varx_{i} \geq \linconst$ 
is the number of variables plus the bit size of representations of the
constant term $\linconst$ and all coefficients~$\lincoeff[i]$, and the
size of a \CP derivation~$\proofstd$ is the sum of the sizes of
all inequalities in~$\proofstd$. We do not know of any degree-like
measure for \CP that would yield relation as that between size and
degree for \PC in \refth{th:ips}. One usually does not distinguish too
carefully between length and size for \CP since
by~\cite{BC96CuttingPlanes} all coefficients in a \CP refutation can
be assumed to have at most exponential size, and are hence
representable with a linear number of bits.

For a partial mapping 
$\funcdescr{\partassign}{\partdomain}{\partrange}$ 
from a domain $\partdomain$ to a range~$\partrange$
we let
$\domainof{\partassign}$ denote  the set of element with an image. 
For 
$\partelem \in \partdomain \setminus \domainof{\partassign}$ 
we write
$\partassign(\partelem)=\partundefined$.
Given a partial assignment or  \introduceterm{restriction}~$\partassign$ 
of variables
$\varx_1, \ldots, \varx_n$
to values in~$\set{0,1}$
and
a polynomial~$\polyp$
or a  linear inequality $\lin$, we denote by
$\restrict{\polyp}{\partassign}$
and
$\restrict{\lin}{\partassign}$
the polynomial and linear inequality obtained from 
$\polyp$ and~$\lin$ by restricting the
variables in the domain of $\partassign$ to the corresponding values
and making obvious syntactic simplifications.
Given a derivation~$\proofstd$ in \PC or \CP, we denote by
$\restrict{\proofstd}{\partassign}$ 
the sequence of restricted polynomials or linear inequalities,
respectively. 
It is straightforward to verify that if $\proofstd$ is a \CP
derivation of an inequality $\lin$ from~$\linset{S}$, then
$\restrict{\proofstd}{\partassign}$ can be viewed (after simple
syntactic manipulations) as a derivation of
$\restrict{\lin}{\partassign}$ from
$\restrict{\linset{S}}{\partassign}$ of at most the same length,
and the same holds for \PC with respect to size and degree.

\subsection{The Graph Colouring Problem}

A \introduceterm{legal $\numcolours$-colouring}
of an undirected graph 
$\graphg = (V, E)$ 
with vertices
$V(\graphg) = V$
and edges
$E(\graphg) = E$
is
a mapping $\funcdescr{\colouringstd}{V}{[k]}$ such that for every edge
$(u,v) \in E$ it holds that 
$\colouringstd(u) \neq \colouringstd(v)$.
The \introduceterm{chromatic number}~$\chomatic(\graphstd)$
of~$\graphstd$ is the smallest $\numcolours$ such that
a legal $\numcolours$-colouring of~$\graphstd$ exists.
In the rest of this paper, colourings will often be assumed to be legal
unless specified otherwise, so we will sometimes omit this prefix when
no misunderstanding can occur. 
Also, it will sometimes be convenient to number the
$\numcolours$~colours 
$0,1, \ldots, \numcolours-1$
instead of
$1, 2, \ldots, \numcolours$,
and we will be fairly relaxed about this issue, implicitly identifying
colours~$0$ and~$\numcolours$ whenever convenient.

Given a graph $\graphstd$ we can encode the 
$\numcolours$-colourability 
problem in a natural way as a system of
polynomial equations over Boolean variables
\begin{subequations}
  \label{eq:colouring}
  \begin{align}
    \label{eq:colouring_defined}
      \sum^{\numcolours}_{j=1} \varx_{v,j} &= 1 
    &&
       \text{$v\in V(\graphstd)$,}
    \\
    \label{eq:colouring_unique}
      \varx_{v,j} \varx_{v,j'} &=  0
    &&
       \text{$v\in V(\graphstd)$, $j \neq j' \in [k]$,}
    \\
    \label{eq:colouring_edges}
      \varx_{u,j}\varx_{v,j} &= 0
    &&
       \text{$(u,v) \in E(\graphstd)$, $j \in[\numcolours]$,}
  \end{align}
\end{subequations}
with the intended meaning that
$\varx_{v,j} = 1 $
if vertex~$v$ has colour
$\colouringstd(v) = j$. 
It is clear that this system of equations has a solution if and only
if the graph~$G$ is \mbox{$\numcolours$-colourable.}

We will also be interested in an alternative algebraic representation
of the $\numcolours$-colouring problem appearing, e.g., in
\ifthenelse{\boolean{conferenceversion}}
{\cite{DLMM08Hilbert,DLMM11ComputingInfeasibility,DLMO09ExpressingCombinatorial}.} 
{\cite{DLMM08Hilbert,DLMO09ExpressingCombinatorial,DLMM11ComputingInfeasibility}.} 
In this encoding every vertex~$v\in V$ has a single associated
variable~$\vary_{v}$ which takes values in
$\set{1,\runity,\runity^{2},\ldots,\runity^{k-1}}$, 
where $\runity$ is a primitive \mbox{$\numcolours$th root} of unity.
The intended meaning is that
$\vary_{v} = \runity^{j}$
if vertex~$v$ has colour $j \in \set{0,1, \ldots, \numcolours-1}$.
The colouring constraints are enforced by the polynomial equations
\begin{subequations}
  \label{eq:colouring_alt}
  \begin{align}
    \label{eq:colouring_alt_defined}
      \vary_{v}^{\numcolours} &= 1
    &&
       \text{$v\in V(\graphstd)$,}
    \\
    \label{eq:colouring_alt_edges}
      \sum^{\numcolours-1}_{j=0}
      {(\vary_{u})}^{j}{(\vary_{v})}^{k-1-j}  
    &= 0 
    &&
       \text{$(u,v)\in E(\graphstd)$,}
  \end{align}
\end{subequations}
where the polynomials live in a polynomial ring over a field
of characteristic that is not a positive number dividing~$\numcolours$.
Clearly, 
Equation~\refeq{eq:colouring_alt_defined}
forces the vertex~$v$ to take some colour.
A moment of thought reveals that Equation~\refeq{eq:colouring_alt_edges}
correctly encodes an edge constraint:
if $\vary_{u}=\runity^{a}$ and $\vary_{v}=\runity^{b}$, then
the sum
evaluates
to $\runity^{b(k-1)}\sum^{\numcolours-1}_{j=0} \runity^{j(a-b)}$, which equals~$0$ when
$ a \neq b$ and $\numcolours \runity^{b(k-1)} \neq 0$~otherwise.
The latter formulation of \kcolouring only makes sense
if the
characteristic of the underlying field $\F$
is either $0$ or a positive integer that does not divide
$\numcolours$.
In this case, we also know that there exists an extension field~$\E$
of~$\F$ that contains a primitive $\numcolours$th root of
unity~$\runity$ \cite[Chapter VI.3]{Lang2005Algebra}. 

A simple but important observation for us is that
the choice of the polynomial encoding
is not too important if we want to study how large
degree is needed in polynomial calculus when proving that some
graph~$\graphg$ is not $\numcolours$-colourable, provided that the
field we are in contains, or can be extended to contain, a
primitive $\numcolours$th~root of unity.

\begin{proposition}
  \label{pr:alternative-encoding-reduction}
  Suppose that 
  \mbox{Equations~\eqref{eq:colouring_alt_defined}--%
    \eqref{eq:colouring_alt_edges}}
  have a polynomial calculus refutation of degree~$\degreestd$ 
  over some field~$\F$ 
  of characteristic that is not a positive number dividing~$k$.
  Then $\F$ can be extended to a field~$\E$ containing a primitive
  $k$th~root of unity~$\runity$, and it holds that
  \mbox{Equations~\eqref{eq:colouring_defined}--%
    \eqref{eq:colouring_edges}}
  have a polynomial calculus refutation over~$\E$
  of degree~$\maxofexpr{2\numcolours,\degreestd}$.
\end{proposition}

\jncomment{EDITED UP TO THIS POINT on Aug 13, 2017.}
\jncomment{We should fill in whatever details that are missing below
  (if any) to make the ``proof sketch'' into a proper proof.}

\mlcomment{Ok, but not for the conference version I hope.
  \textbf{Jakob:} Sure, but now we are writing the full-length
  version, right? ;-)
}

\begin{proof}
  By the assumption on the characteristic of $\F$, we already argued
  that there exists some extension field~$\E$ of~$\F$ that contains
  a primitive $\numcolours$th root of unity~$\runity$.
  We plan to translate a polynomial calculus refutation~$\proofstd$ of
  \mbox{Equations~\eqref{eq:colouring_alt_defined}--%
    \eqref{eq:colouring_alt_edges}}, into a refutation of
  \mbox{Equations~\eqref{eq:colouring_defined}--%
    \eqref{eq:colouring_edges}}, and the first step of this process is
  to apply the linear substitutions
  \begin{equation}
    \label{eq:substitution_prelims}
    \vary_{v} \mapsto \sum^{\numcolours}_{j=1} \varx_{v,j} \runity^{j}
  \end{equation}
  to all variables in all polynomials in $\proofstd$ to obtain
  a new sequence of polynomials~$\proofstd'$ in
  variables~$\varx_ {v,j}$. These substituted polynomials in
  $\proofstd'$ have coefficients in $\E$ and we will use them to form
  the skeleton of our new refutation. In order to turn $\proofstd'$
  into a refutation of \mbox{Equations~\eqref{eq:colouring_defined}--%
    \eqref{eq:colouring_edges}} we are going to
  \begin{itemize}
    \item show that for every application of a derivation rule in
    $\proofstd$, it is possible to derive the corresponding
    substituted consequence from the substituted premises with no
    increase in the degree;
    \item show that the substituted axioms in $\proofstd'$ can be
    derived from \mbox{Equations~\eqref{eq:colouring_defined}--%
      \eqref{eq:colouring_edges}} in degree $2k$.
  \end{itemize}
  The first item is almost immediate. All applications of the linear
  combination rule in~$\proofstd$ remain valid in~$\proofstd'$, since
  the substitution is a linear operator.
  When $\vary_{u} p$ is derived from $p$ in $\proofstd$ by the means of
  an application of the multiplication rule, in the new refutation we
  need to derive $\sum^{\numcolours}_{j=1} \varx_{v,j} \runity^{j} p'$
  from $p'$, where $p'$ is the substituted version of $p$. To do that
  it is sufficient to derive each $ \varx_{v,j} p'$ from $p'$ and then
  take a linear combination. Notice that the degree of this derivation
  is the same as in $\proofstd$.
   
  It remains to argue is that the substituted versions of the initial
  \mbox{axioms~\eqref{eq:colouring_alt_defined}--%
    \eqref{eq:colouring_alt_edges}} in~$\proofstd$ can be derived
  from the
  \mbox{axioms~\eqref{eq:colouring_defined}--\eqref{eq:colouring_edges}}
  available to the new refutation.
  We first derive the substituted version of the axiom
  $\vary_{v}^{\numcolours} - 1$ (which is just
  axiom~\eqref{eq:colouring_alt_defined}
  with the additive~$1$ moved to the left side), namely
  \begin{equation}
    \label{eq:defined-axiom-substituted_unexpanded}
    {\left(\sum^{\numcolours}_{j=1} \varx_{v,j} \runity^{j} \right)}^{\numcolours} - 1
    \eqcomma
  \end{equation}
  which after expansion becomes
  \begin{equation}
    \label{eq:defined-axiom-substituted}
    \sum^{\numcolours}_{j=1} \varx_{v,j}^{\numcolours} \runity^{j\numcolours} - 1 + \polyq
    \eqperiod
  \end{equation}
  where each monomial in $\polyq$ contains some factor on the
  form~$\varx_{v,j}\varx_{v,j'}$ \mbox{for $j \neq j'$}.
  We use that $\runity^{k}$ is $1$ and we
  rewrite~\refeq{eq:defined-axiom-substituted} as
  \begin{equation}
    \sum^{\numcolours}_{j=1} \sum_{\ell=0}^{\numcolors-2}\varx_{v,j}^{\ell}(\varx_{v,j}^{2} - \varx_{v,j})
    +
    \left( \sum^{\numcolours}_{j=1} \varx_{v,j} - 1  \right) + \polyq
    \eqcomma
  \end{equation}
  The $\polyq$ part is derivable from
  axioms~\eqref{eq:colouring_unique},
  and the rest can be derived from boolean axioms and
  axiom~\refeq{eq:colouring_defined}.

  Now we focus on the substituted
  axiom~\eqref{eq:colouring_alt_edges}, which is
  \begin{equation}
    \label{eq:edge-axiom-substituted_unexpanded}
    \sum^{\numcolours-1}_{j=0} 
    {\left(\sum^{\numcolours}_{a=1} \varx_{u,a} \runity^{a} \right)}^{j}
    {\left(\sum^{\numcolours}_{b=1} \varx_{v,b} \runity^{b} \right)}^{k-1-j}
    \eqperiod
  \end{equation}
  It will be more convenient for us to first derive the polynomial of
  degree $2\numcolours-1$ 
  \begin{equation}
    \label{eq:edge-axiom-substituted_padded}
    \sum^{\numcolours-1}_{j=0} 
    {\left(\sum^{\numcolours}_{a=1} \varx_{u,a} \runity^{a} \right)}^{\numcolours+j}
    {\left(\sum^{\numcolours}_{b=1} \varx_{v,b} \runity^{b} \right)}^{2\numcolours-1-j}
    \eqcomma
  \end{equation}
  and then show that it is equivalent to
  \eqref{eq:edge-axiom-substituted_unexpanded} using the
  equation~\eqref{eq:defined-axiom-substituted_unexpanded}
  proved above.
  After expansion polynomial~\eqref{eq:edge-axiom-substituted_padded} becomes
  \begin{equation}
    \label{eq:variable_substitution_edges_nonmult}
    \sum^{\numcolours-1}_{j=0} 
    \sum^{\numcolours}_{a=1} 
    \sum^{\numcolours}_{b=1} \varx_{u,a}^{\numcolours+j} \varx_{v,b}^{(2\numcolours-1-j)}
    \runity^{(\numcolours+j)a}\runity^{(2\numcolours-1-j)b}
    \quad
    +
    \quad
    \polyq'\eqcomma
  \end{equation}
  where each monomial in $\polyq'$ contains either some factor
  $\varx_{u,a}\varx_{u,a'}$ for $a\neq a'$ or some
  $\varx_{v,b}\varx_{v,b'}$ for $b\neq b'$.
  Since $\numcolours+j$ and $2\numcolours-1+j$ are both greater than
  zero for $j \in [0,\numcolours-1]$, the Boolean axioms can be used
  to prove the equivalence of each
  $\varx_{u,a}^{\numcolours+j} \varx_{v,b}^{(2\numcolours-1-j)}$ with
  $\varx_{u,a}\varx_{v,b}$. This fact, together with the fact that
  $\runity^{\numcolours}=1$, allows us to reduce to the derivation of
  \begin{equation}
    \label{eq:variable_substitution_edges}
    \sum^{\numcolours-1}_{j=0} 
    \sum^{\numcolours}_{a=1} 
    \sum^{\numcolours}_{b=1} \varx_{u,a} \varx_{v,b}
    \runity^{ja}\runity^{(k-1-j)b}
    \quad
    +
    \quad
    \polyq'\eqperiod
  \end{equation}
  All monomial in~$\polyq'$ are
  derivable from axioms~\eqref{eq:colouring_unique}.
  To derive first summand we change the order of the
  summation and 
  split it into two parts, depending on whether
  $a=b$,
  to obtain
  \begin{equation}
    \label{eq:variable_substitution_edges_split}
    \sum^{\numcolours}_{a=1} 
    \varx_{u,a} \varx_{v,a}
    \sum^{\numcolours-1}_{j=0} 
    \runity^{ja}\runity^{(k-1-j)a}
    \quad
    +
    \quad
    \sum^{\numcolours}_{a=1} 
    \sum^{\numcolours}_{b=1,b \neq a}
    \varx_{u,a} \varx_{v,b} 
    \underbrace{%
      \left(
       \runity^{b(k-1)}
      \sum^{\numcolours-1}_{j=0}
      \runity^{j(a-b)}
      \right)
    }_{\text{equal to $0$ for $a \neq b$}}\eqperiod
  \end{equation}
  The axioms~\eqref{eq:colouring_edges} allow to derive the first part of
  \eqref{eq:variable_substitution_edges_split}, while the second part
  is identically zero.
  In conclusion we have shows how to derive substituted axioms
  \eqref{eq:defined-axiom-substituted_unexpanded} and
  \eqref{eq:edge-axiom-substituted_unexpanded} in degree
  $\leq 2\numcolours$ and therefore we have
  concluded the translation of refutation $\proofstd$.
\end{proof}

For later use, we note that we can also
encode the $\numcolours$-colourability problem for a
graph~$\graphg$ as a system of linear inequalities 
\begin{subequations}
  \begin{align}
    \label{eq:colouring_defined_cp}
      \sum^{\numcolours}_{j=1} \varx_{v,j} &\geq 1 
    &&
       \text{$v\in V(\graphstd)$,}
    \\
    \label{eq:colouring_unique_cp}
      \varx_{v,j} + \varx_{v,j'} &\leq 1
    &&
       \text{$ v \in V(\graphstd)$, $j \neq j' \in[\numcolours]$,}
    \\
    \label{eq:colouring_edges_cp}
      \varx_{u,j} + \varx_{v,j} &\leq 1
    &&
       \text{$(u,v)\in E(\graphstd)$, $j \in[\numcolours]$,}
  \end{align}
  in a format amenable to cutting planes reasoning.
\end{subequations}

\section{Worst-Case Lower Bound for Polynomial Calculus}
\label{sec:pc-lower-bound}

We now show how to explicitly construct a family of graphs which are
not \kcolourable but for which polynomial calculus proofs of this fact
(over any field)
require degree linear in the number of vertices in the graphs.
We do this in three steps:
\begin{enumerate}       \itemsep=0pt
\item 
  First, we show how to reduce 
  instances of
  \introduceterm{functional pigeonhole principle (FPHP) formulas} 
  defined over bipartite graphs of bounded degree  to
  graph colouring instances so that there is a one-to-one mapping of
  pigeons to holes if and only if the graph is \kcolourable.
\item
  Then we show that polynomial calculus is able to carry out this
  reduction in constant degree, so that a low-degree PC proof of graph
  non-colourability can be used to obtain a low-degree refutation of
  the corresponding FPHP instance.
\item
  Finally, we appeal to a linear lower bound on degree for refuting
  FPHP instances over bipartite expander graphs
  from~\cite{MN15GeneralizedMethodDegree}. 
\end{enumerate}

Let us start by giving a precise description of our functional
pigeonhole principle instances.
We have a set of pigeons~$\pigeonset$ which want to fly into a set of
holes~$\holeset$, with each pigeon flying into exactly one hole in a
one-to-one fashion. However, the choices of holes for the pigeons are
constrained, so that pigeon~$\pigeonindex$ can fly only to the holes 
in~$\holesforpigeon{\pigeonindex} \subseteq \holeset$,
where we have
$\setsize{\holesforpigeon{\pigeonindex}} = \numcolours$.
If we use variables
$p_{\pigeonindex,\holeindex}$
to denote that 
pigeon~$\pigeonindex$ flies into hole~$\holeindex$, 
we can write the constraints on such a mapping as a set of polynomial
equations 
\begin{subequations}
  \begin{align}
    \label{eq:fphp_defined}
      \sum_{\holeindex \in \holesforpigeon{\pigeonindex}}
      p_{\pigeonindex,\holeindex} 
      &= 1
    &&
       \text{$\pigeonindex \in \pigeonset$,}
    \\
    \label{eq:fphp_unique}
      p_{\pigeonindex,\holeindex}p_{\pigeonindex,\holeindex'} 
      &= 0
    &&
       \text{$\pigeonindex \in I$, 
       $\holeindex \neq \holeindex' \in \holesforpigeon{\pigeonindex}$.}
    \\
    \label{eq:fphp_holes}
      p_{\pigeonindex ,\holeindex}p_{\pigeonindex ',\holeindex} 
           &= 0
    &&
       \text{$\pigeonindex \neq  \pigeonindex ' \in \pigeonset$,
       $\holeindex \in \holesforpigeon{\pigeonindex} \intersection
       \holesforpigeon{\pigeonindex'}$.}
  \end{align}
\end{subequations}
Note that an instance encoded by 
\mbox{Equations~\refeq{eq:fphp_defined}--\refeq{eq:fphp_holes}}
can also be naturally viewed as a bipartite graph~$\bipartstd$ with
left vertex set~$\pigeonset$, right vertex set~$\holeset$, and edges
from each $\pigeonindex \in  \pigeonset$ to all $\holeindex \in
\holesforpigeon{\pigeonindex}$.  In what follows, we will mostly
reason about FPHP instances in terms of their representations as
bipartite graphs.

In the standard setting, we let
$\pigeonset = [n]$
and
$\holeset = [n-1]$
for some~$n \in \N$,
in which case the collection of  
\mbox{constraints~\refeq{eq:fphp_defined}--\refeq{eq:fphp_holes}}
is clearly unsatisfiable.
Nevertheless, it was shown in~\cite{MN15GeneralizedMethodDegree} 
that if the underlying bipartite graph is a so-called
\introduceterm{boundary expander}, 
then any \PC refutation of
\mbox{Equations~\refeq{eq:fphp_defined}--\refeq{eq:fphp_holes}}
requires $\bigomega{n}$ degree and thus, by \refth{th:ips},
exponential size.
The \introduceterm{boundary} of some set $X$ of vertices in a (either
simple or bipartite) graph $\graphstd$ is the set of vertices in
$V(\graphstd)\setminus X$ that have exactly one neighbor in $X$.
Informally, a (bipartite) boundary expander is (bipartite) graph for
which every set of (left) vertices of reasonable size has
large boundary.
For our results we do not need to go into the technical details of
the lower bound for FPHP. It suffices to use the following claim
as a black box.

\mlcomment{I fixed the claim from \cite{MN15GeneralizedMethodDegree},
  in the sense that I don't attribute anymore to you the existence of
  the bipartite family, but only that given the right family the
  result holds. There are two points that must be addressed and
  i leave you to decide how to do that.
  \begin{itemize}
    \item We could say that the paper does not need to use the right
    degree bound;
    \item the paper does not claim a lower bound for FPHP(B) in case
    of boundary expansion with additive error.
  \end{itemize}}

\begin{theorem}[By the proof of \cite{MN15GeneralizedMethodDegree}]
  \label{thm:fphp}
  For any integer $\numcolours \geq 3$, consider a family of bipartite
  graphs $\{\bipartstd_{n}\}_{n\in \N}$ with
  \begin{itemize}
    \item $n$~vertices on the left side, $n-1$ vertices on the right
    side, left degree~$\numcolours$, and right
    degree~$\bigoh{\numcolours}$,
    \item there are universal constants $\alpha,\delta$ so that for any
    $n\in\N$ and any set $I$ of at most $\alpha n$ vertices on the left
    side of $\bipartstd_{n}$, the size of the boundary of $I$ is at
    least $\delta|I|-1$.
  \end{itemize}
  Any polynomial calculus refutation of the
  \mbox{constraints~\refeq{eq:fphp_defined}--\refeq{eq:fphp_holes}}
  corresponding to $\bipartstd_{n}$ requires \mbox{degree
    $\bigomega{n}$.}
\end{theorem}

\jncomment{\textbf{Looking at what happens next, it seems that you also need
    bounded right degree}, and so I added this to the theorem statement.
  We should probably also claim/make clear in running text that
  left-regular
  bipartite boundary expanders with bounded right degree also exist.
  (Which is probably true, but I honestly don't know a reference off
  the top of my head --- maybe we should simply try proof by
  intimidation? ;-) )}
\mlcomment{Yes bounded right degree $O(\numcolours)$ is necessary to
  have bounded 
  degree in the final graph. If it was not stated then it was
  a mistake. I always state it when I explain the result. }
\mlcomment{Isn't true that an $n$ by $n$ $d$-regular random graph is
  an expander with high probability? Or maybe a non bipartite expander
  which is ``doubled'', by taking two copies and adding cross edges
  where the original edges were. 
  Then it is sufficient to remove one vertex on one side and
  reallocate the $d$ edges at random without collision.
  \textbf{I believe the answer is ``yes'' in both cases, but we would
    need a reference for this. -Jakob}
}

To be precise, 
the lower bound in \refth{thm:fphp} was proven for a slightly
different encoding of
\mbox{Equations~\refeq{eq:fphp_defined}--\refeq{eq:fphp_holes}}---%
namely the one obtained from the natural translation of CNF formulas
into polynomial equations---%
but the two encodings imply each other and can be used to derive each
other in degree $\bigoh{\numcolours}$ by the implicational
completeness of polynomial calculus. Hence, the lower bound holds for
both encodings.

\jncomment{Any particular reason why we don't simply use the CNF
  encoding? \\
  This claim is probably also true by intimidation --- at least if
  functionality axioms are used ;-) --- but I didn't check how to do
  the derivations. Did you?
}

\mlcomment{This paper is not about CNFs. If we start mentioning the
  CNF encoding we end up with two different encoding for truth
  and falsity. That would be confusing and not so neat.
  Regarding the translation between the two encoding: each equation to prove
  in one formulation follows by a bunch of equations in the other formulation
  that involve just $O(\numcolors)$ variables. Therefore each bit of the
  translation can be done in degree $O(\numcolors)$ by implicational completeness.}

We proceed to describe the reduction from functional pigeonhole
principle instances to graph colouring instances.
Our starting point is an FPHP instance 
on a bipartite graph~$\bipartstd$
with pigeons $\pigeonset = [n]$ and holes~$\holeset$
where every pigeon has exactly
$d_{\pigeonset}=\numcolours$ holes to choose from and every hole can
take
$\bigoh{\numcolours}$~pigeons; \ie the bipartite graph~$\bipartstd$
is left-regular of degree~$\numcolours$ and has right
degree~$\bigoh{\numcolours}$. 
Based on this instance we construct a graph 
$\graphstd = \graphstd(\bipartstd)$ 
such that $\graphstd$ is \kcolourable if and only if the functional
pigeonhole principle on $B$ is satisfiable.

By way of overview, the graph~$\graphstd(\bipartstd)$ has $n$~special
vertices corresponding to the pigeons, and the colours of these
vertices encode how the pigeons are mapped to holes.
For every pair of pigeons~$i,i'$ that can be mapped to the same
hole~$j$ we add a gadget that forbids the colouring of the pigeon
vertices~$i$ and~$i'$ that corresponds to them being mapped to
hole~$j$.  These gadgets have a couple of pre-coloured vertices, but
we eliminate such pre-colouring by adding one more simple gadget.

In more detail, 
the main idea behind the reduction is to view the choices 
$\holesforpigeon{\pigeonindex}$
for each pigeon
$\pigeonindex \in \pigeonset$
as taking the first, second, \ldots, $\numcolours$th edge.
We fix
an arbitrary enumeration of the elements of~$\holesforpigeon{\pigeonindex}$ 
for each $\pigeonindex \in \pigeonset$, associating distinct numbers
$1, 2, \ldots, \numcolours$
to the edges out of the vertex~$\pigeonindex$ in~$\bipartstd$. 
We say that
\introduceterm{pigeon $\pigeonindex$ flies to hole~$\holeindex$ using
  its $\colstd$th edge}
if the edge connecting pigeon~$\pigeonindex$ to hole~$\holeindex$ is
labelled by $\colstd \in [\numcolours]$, and use the notation
$\flightviaedge{\pigeonindex}{\colstd}$
for this (suppressing the information about the hole~$\holeindex$).
Pigeon~$\pigeonindex$ taking the $\colstd$th edge corresponds to the
special $i$th~pigeon vertex being coloured with colour~$\colstd$.

Consider two distinct pigeons 
$\pigeonindex \neq \pigeonindex' \in \pigeonset$ 
and a hole
$\holeindex \in 
\holesforpigeon{\pigeonindex} \intersection \holesforpigeon{\pigeonindex'}$.
If pigeon $\pigeonindex{}$ flies to hole~$\holeindex$ using its
$\colstd$th~edge and pigeon~$\pigeonindex{}'$ flies to
hole~$\holeindex$ using its $\colstd'$th edge, then 
the translation of the  injectivity constraint~\eqref{eq:fphp_holes}
expressed in terms of \kcolouring{}s is that
vertices $\pigeonindex$ and~$\pigeonindex{}'$ cannot be
simultaneously coloured by colours~$\colstd$ and~$\colstd'$,
respectively.

\begin{figure}[tp]
  \centering
  \begin{subfigure}[b]{0.52\textwidth}
    \centering
    \includegraphics{colouringgadget.11}
  \ifthenelse{\boolean{conferenceversion}}
  {\caption{Forbidding
      $\flightviaedge{\pigeonindex}{\colstd}$
      and
      $\flightviaedge{\pigeonindex'}{\colstd}$.}}
  {\caption{Gadget ruling out
      $\flightviaedge{\pigeonindex}{\colstd}$
      and
      $\flightviaedge{\pigeonindex'}{\colstd}$.}}
    \label{fig:gadgetsame}
  \end{subfigure}%
  \hfill
  \begin{subfigure}[b]{0.42\textwidth}
    \centering
    \includegraphics{colouringgadget.12}
  \ifthenelse{\boolean{conferenceversion}}
  {\caption{Forbidding
      $\flightviaedge{\pigeonindex}{\colstd}$
      and
      $\flightviaedge{\pigeonindex'}{\colstd'}$
      for
      $\colstd \neq \colstd'$.}}
  {\caption{Gadget ruling out
      $\flightviaedge{\pigeonindex}{\colstd}$
      and
      $\flightviaedge{\pigeonindex'}{\colstd'}$
      for
      $\colstd \neq \colstd'$.}}
    \label{fig:gadgetdiff}
  \end{subfigure}
  \caption{Injectivity constraint gadgets 
    $\diffcolgadgetstd$ for $\numcolours=4$.}
  \label{fig:gadget}
\end{figure}

Let us now give a precise description of
the graph gadgets we 
employ
to enforce
such injectivity constraints.
These will be 
partially pre-coloured graphs
$\diffcolgadgetstd$
as depicted in 
\reftwofigs{fig:gadgetsame}{fig:gadgetdiff}. 
The gadget
constructions start with two disjoint
$\numcolours$\nobreakdash-cliques for pigeons~$\pigeonindex$
\ifthenelse{\boolean{conferenceversion}}
{and~$\pigeonindex'$, which we will refer to as the left and right
  cliques, respectively.}
{and~$\pigeonindex'$.  For the 
sake of exposition we call the former the left clique and the latter
the right clique, although the construction fully symmetric.}
We refer to the vertices in the left clique as 
$\lcvertex{1}, \ldots, \lcvertex{\numcolours}$ numbered in
a clockwise fashion starting with the first vertex at the bottom, and
in a symmetric fashion the vertices in the right clique 
are referred to as
$\rcvertex{1}, \ldots, \rcvertex{\numcolours}$ 
numbered anti-clockwise starting at the bottom.

To the first vertex~$\lcvertex{1}$ in the left
\mbox{$\numcolours$-clique} we connect the vertex~$\pigeonindex$. To vertices 
$\lcvertex{2}, \ldots, \lcvertex{\numcolours-1}$ 
we connect a new vertex pre-coloured with colour~$\colstd$. 
For the right \mbox{$\numcolours$-clique} we do a similar construction:
to the first vertex~$\rcvertex{1}$ 
we connect the vertex $\pigeonindex'$ and to the next
$\numcolours-2$ vertices 
$\rcvertex{2}, \ldots, \rcvertex{\numcolours-1}$ 
we connect a new vertex pre-coloured with
colour~$\colstd'$. 

The final step of the construction depends on whether
$\colstd = \colstd'$ or not. If $\colstd = \colstd'$, 
then we add an edge between the final two vertices 
$\lcvertex{\numcolours}$
and~$\rcvertex{\numcolours}$
in the cliques.
\mbox{If $\colstd \neq \colstd'$}, then we instead merge
these two vertices into a single vertex 
as shown in \reffig{fig:gadgetdiff}.
We want to stress that except for $\pigeonindex$ and~$\pigeonindex'$
all vertices in the construction are new vertices that do not occur in
any other gadget. 
Let us collect for the record some properties of this gadget construction.

\begin{claim}
  \label{clm:gadget}
  The pre-coloured graph gadget
  $\diffcolgadgetstd$
  has the following properties:
  \begin{enumerate} \itemsep=0pt
  \item 
    $\diffcolgadgetstd$
    has $\bigoh{\numcolours}$ vertices.
    
  \item 
    $\diffcolgadgetstd$
    has two pre-coloured vertices of degree $\bigoh{\numcolours}$.
    
  \item\label{item:complete} 
    For every 
    $(\colalt,\colalt') \neq (\colstd,\colstd')$ 
    there is
    a legal \kcolouring~$\vcolouring$ of $\diffcolgadgetstd$ extending the
    pre-colouring and 
    satisfying
    $\vcolour{\pigeonindex} = \colalt$ and
    $\vcolour{\pigeonindex'} = \colalt'$.
    No such legal \kcolouring of~$\diffcolgadgetstd$ exists for
    $(\colalt,\colalt')=(\colstd,\colstd')$.

  \end{enumerate}
\end{claim}

\begin{proof}
  The first two properties obviously hold by construction.
    
  To prove Property~\ref{item:complete}, let us focus on the left
  clique in either of the two variant of the gadget.
  If  $\vcolour{\pigeonindex} = \colstd$, then clearly 
  vertex~$\lcvertex{1}$ in
  the left clique cannot take colour~$\colstd$. Since the pre-coloured
  vertex connected to vertices 
  \mbox{$\lcvertex{2}, \ldots, \lcvertex{\numcolours-1}$}
  of the clique also has colour~$\colstd$, and since any legal
  colouring must use all available colours for the clique, this forces
  $\vcolour{\lcvertex{\numcolours}} = \colstd$. 
  If $\vcolour{\pigeonindex} \neq \colstd$, 
  however, then we can colour vertex~$\lcvertex{1}$ with colour~$\colstd$, 
  and then choose any permutation of the remaining colours for the
  other vertices in the left clique, giving the 
  vertex~$\lcvertex{\numcolours}$
  at least two distinct colours to choose 
  \ifthenelse{\boolean{conferenceversion}}
  {between.}
  {between
  after the other clique vertices 
  \mbox{$\lcvertex{2}, \ldots, \lcvertex{\numcolours-1}$}
  connected to the pre-coloured vertex have been coloured.}

  Consider now the case
  $\colstd = \colstd'$, 
  so that we have the graph gadget
  $\diffcolgadgetsamestd$ in 
  \reffig{fig:gadgetsame}.
  By symmetry, if
  $\vcolour{\pigeonindex'} = \colstd'$, 
  then this forces
  $\vcolour{\rcvertex{\numcolours}} = \colstd$, 
  but there are at least two choices for the colour 
  of~$\rcvertex{\numcolours}$
  if $\vcolour{\pigeonindex'} \neq \colstd'$.
  It follows that if
  $\flightviaedge{\pigeonindex}{\colstd}$
  and~$\flightviaedge{\pigeonindex'}{\colstd}$,
  then 
  vertices~$\lcvertex{\numcolours}$ and~$\rcvertex{\numcolours}$
  both have to get the same
  colour~$\colstd$ to avoid conflicts in the left and right
  $\numcolours$\nobreakdash-cliques, respectively, 
  which causes a conflict 
  along the edge $(\lcvertex{\numcolours}, \rcvertex{\numcolours})$.
  As long as one of $\pigeonindex$ and~$\pigeonindex'$ is assigned a
  colour other than~$\colstd$, however, 
  $\diffcolgadgetsamestd$~can be legally \kcoloured.
  \ifthenelse{\boolean{conferenceversion}}
  {For
    $\colstd \neq \colstd'$
    we reason analogously but use instead the graph gadget
    $\diffcolgadgetstd$ in 
    \reffig{fig:gadgetdiff}.}
  {

  For
  $\colstd \neq \colstd'$
  we instead have the graph gadget
  $\diffcolgadgetstd$ in 
  \reffig{fig:gadgetdiff}.
  Following the same line of reasoning as above, if
  $\flightviaedge{\pigeonindex}{\colstd}$, 
  then the left $\numcolours$\nobreakdash-clique  
  forces~$\lcvertex{\numcolours}$ 
  to take colour~$\colstd$, and if
  $\flightviaedge{\pigeonindex'}{\colstd'}$,
  then the right $\numcolours$\nobreakdash-clique forces
  $\rcvertex{\numcolours} = \lcvertex{\numcolours}$
  to take colour~$\colstd' \neq \colstd$, which is a
  conflict. If
  $\flightviaedge{\pigeonindex}{\colstd}$
  but~$\flightnotviaedge{\pigeonindex'}{\colstd'}$, 
  however, then
  $\lcvertex{\numcolours}$
  can first be coloured with colour~$\colstd$ to satisfy
  the left $\numcolours$\nobreakdash-clique, after which the
  \kcolouring of the right  $\numcolours$\nobreakdash-clique can be
  completed taking 
  $\vcolour{\lcvertex{\numcolours}}$ 
  into consideration,
  and the reasoning is symmetric if
  $\flightviaedge{\pigeonindex'}{\colstd'}$
  but~$\flightnotviaedge{\pigeonindex}{\colstd}$.
  The case when both $\flightnotviaedge{\pigeonindex}{\colstd}$
  and~$\flightnotviaedge{\pigeonindex'}{\colstd'}$ is dealt with in a
  similar way.
  This establishes the claim.}
\end{proof}

We write 
$\alldiffcolgadgets = \alldiffcolgadgets(\bipartstd)$ 
to denote the graph consisting of
the union of all gadgets
$\diffcolgadgetstd$
for all
$\pigeonindex \neq \pigeonindex' \in \pigeonset$
and all
$\colstd,\colstd'$
such that if 
pigeon~$\pigeonindex$ uses  its $\colstd$th~edge and
pigeon~$\pigeonindex'$ uses  its $\colstd'$th~edge in~$\bipartstd$,
then they both end up in the same hole~$\holeindex \in \holeset$. 
All vertices corresponding to  pigeons
$\pigeonindex \in \pigeonset$
are shared between gadgets~$\diffcolgadgetstd$
in~$\alldiffcolgadgets$, but apart from this all subgraphs
$\diffcolgadgetstd$ are vertex-disjoint.  We next state some
properties of~$\alldiffcolgadgets$.

\begin{lemma}
  \label{lem:completion}
  Consider an FPHP instance encoded by
  \mbox{Equations~\refeq{eq:fphp_defined}--\refeq{eq:fphp_holes}}
  for a left-regular bipartite graph with left degree
  $d_{\pigeonset} = \numcolours$
  and bounded right degree
  $d_{\holeset} = \bigoh{\numcolours}$, 
  and let 
  $\alldiffcolgadgets$ 
  be the partially \kcoloured graph
  obtained as described above. 
  Then
  $\alldiffcolgadgets$ 
  has 
  $\Bigoh{\numcolours^{4} \setsize{\pigeonset}}$ 
  vertices and maximal vertex degree~$\bigoh{\numcolours^{2}}$, 
  and the number of pre-coloured vertices is
  $\Bigoh{\numcolours^{2} \setsize{\pigeonset}}$. 
  Furthermore, 
  the partial \kcolouring of $\alldiffcolgadgets$ can be extended to a
  complete, legal \kcolouring 
  of~$\alldiffcolgadgets$
  if and only if there is a way to map each pigeon
  $\pigeonindex \in \pigeonset$ to some hole
  $\holeindex \in \holeset$
  without violating any constraint 
  in~\mbox{\refeq{eq:fphp_defined}--\refeq{eq:fphp_holes}}.
\end{lemma}

\begin{proof}
  Without loss of generality we can assume that
  $\setsize{\holeset} \leq \numcolours \setsize{\pigeonset}$
  \ifthenelse{\boolean{conferenceversion}}
  {(otherwise there are holes that cannot be used by any pigeon).}
  {(otherwise there are holes that cannot be used by any pigeon and
  that can thus be discarded).}
  Each gadget 
  $\diffcolgadgetstd$
  has $\bigoh{\numcolours}$ vertices and there are   at most
  $(d_{\holeset})^{2} = \Bigoh{\numcolours^{2}}$ 
  distinct pairs of pigeons that can fly to any single
  hole~$\holeindex$,
  meaning that we have a total of at most
  $\Bigoh{\numcolours^{2} \setsize{\holeset} }$ 
  injectivity constraint gadgets
  $\diffcolgadgetstd$.
  Therefore, by a crude estimate   
  $\alldiffcolgadgets$ has at most
  $\Bigoh{\numcolours^{4}\setsize{\pigeonset}}$ vertices 
  in total. 

  By Claim~\ref{clm:gadget} at most
  $\Bigoh{\numcolours^{2} \setsize{\pigeonset} }$ 
  vertices in~$\alldiffcolgadgets$ 
  are pre-coloured.
  Each pigeon vertex labelled by $\pigeonindex \in \pigeonset$
  is involved in at most
  $d_{\pigeonset} d_{\holeset} = \Bigoh{\numcolours^{2}}$
  injectivity constraint gadgets, so such vertices have degree
  $\Bigoh{\numcolours^{2}}$, 
  while all other vertices have degree  $\bigoh{\numcolours}$.
  
  For any complete colouring of~$\alldiffcolgadgets$
  extending the pre-colouring,   the colours 
  $\vcolour{\pigeonindex} = \colstd_\pigeonindex$
  assigned to pigeon vertices
  $\pigeonindex \in  \pigeonset$
  define a mapping from pigeons to holes via the chosen
  edges~$\colstd_\pigeonindex$. 
  It follows from  Claim~\ref{clm:gadget} that this colouring is
  legal only if pigeons are mapped to holes in a one-to-one fashion,
  which implies that
  \mbox{Equations \refeq{eq:fphp_defined}--\refeq{eq:fphp_holes}}
  are satisfiable. In the other direction, for any one-to-one mapping
  of pigeons to holes we can colour vertex~$\pigeonindex$ by the
  colour~$\colstd_i$ corresponding to the edge it uses to fly to its
  hole, and such a colouring can be combined with the pre-colouring
  complete, 
  to produce a legal \kcolouring.
\end{proof}

To 
finalize 
our reduction we need to get rid of the pre-coloured
vertices in~$\alldiffcolgadgets$.
To this end, we first make the following observation.
Recall that for every every pigeon~$\pigeonindex \in \pigeonset$ we
fixed an enumeration of the edges to holes
$\holeindex \in \holesforpigeon{\pigeonindex}$ in~$\bipartstd$, so
that the choice of an edge corresponds to the choice of a colour.
Suppose we 
apply some arbitrary but fixed permutation $\permsigma$
on~$[\numcolours]$ to all such enumerations for the pigeons
$\pigeonindex \in \pigeonset$. Clearly, this does not change the
instance in any significant way. 
If it was the case before that
pigeon~$\pigeonindex$ and~$\pigeonindex'$
could not simultaneously take the $\colstd$th and $\colstd'$th~edges,
respectively, then now these pigeons cannot simultaneously take the
$\permsigma(\colstd)$th
and
$\permsigma(\colstd')$th~edges, respectively.
In other words,  \reflem{lem:completion} 
is invariant with respect to any permutation of the
colours~$[\numcolours]$, and we could imagine the reduction as
first picking some permutation~$\permsigma$ and then
constructing $\alldiffcolgadgets$ \wrt this permutation.

A simple way of achieving this effect would be to construct a separate
``pre-colouring $\numcolours$\nobreakdash-clique'' consisting
of $\numcolours$ special vertices
$\clrvertex{1}, \ldots, \clrvertex{\numcolours}$, and then identify all vertices
in~$\alldiffcolgadgets$ pre-coloured with colour~$\colstd$
with the vertex~$\clrvertex{\colstd}$.
It is not hard to see that the resulting graph would be \kcolourable
if and only if the pre-colouring of~$\alldiffcolgadgets$ could be
extended to a complete, legal \kcolouring, and using
\reflem{lem:completion} we would obtain a valid reduction from
the functional pigeonhole principle to graph \kcolouring. However, the
final graph would have degree 
$\Bigomega{\numcolours^{3}\setsize{\pigeonset}}$, 
and we would like to obtain a graph of bounded degree.

\begin{figure}
    \centering
    \includegraphics{colouringgadget.13}
    \ifthenelse{\boolean{conferenceversion}}
    {\caption{Pre-colouring gadget with vertices to be
      identified with  the 
      pre-coloured vertices in $\alldiffcolgadgets$.}}
    {\caption{Pre-colouring gadget with uncoloured vertices to be
      identified with  the 
      pre-coloured vertices in $\alldiffcolgadgets$.}}
    \label{fig:gadgetseq}
\end{figure}

To keep the 
\ifthenelse{\boolean{conferenceversion}}
{vertex degrees} 
{graph vertex degree} 
independent
of 
the size~$\setsize{\pigeonset}$ of the left-hand side of the FPHP
bipartite graph~$\bipartstd$, we instead construct a pre-colouring
gadget using a slight modification of the above idea.  Consider a set
$\set{\clrvertex{1},\clrvertex{2},\ldots,\clrvertex{\precolgadgetsize}}$
of new vertices, for $\precolgadgetsize$ to be fixed later. For every segment of
$\numcolours$~consecutive vertices
$\set{\clrvertex{\clrvindex},\clrvertex{\clrvindex+1},\ldots,\clrvertex{\clrvindex+\numcolours-1}}$
we add all edges
$\Setdescr
{(\clrvertex{\colstd}, \clrvertex{\colstd'})}
{\colstd \neq \colstd' \in 
  \set{\clrvindex, \clrvindex+1, \ldots, \clrvindex+\numcolours-1}}$
so that they form a $\numcolours$\nobreakdash-clique as illustrated in 
\reffig{fig:gadgetseq}
(where as in
\reffig{fig:gadget} we have~$\numcolours=4$). 
Next, we go through all the pre-coloured
vertices in $\alldiffcolgadgets$: if a vertex should be pre-coloured
by~$\colstd$, then
we identify it with the first vertex~$\clrvertex{\clrvindex}$ such that 
$\clrvindex \equiv \colstd \pmod{\numcolours}$ and such
that $\clrvertex{\clrvindex}$ has not already been 
\ifthenelse{\boolean{conferenceversion}}
{used at a previous step.}
{assigned at a previous step to some other pre-coloured vertex.} 
If we choose 
$\precolgadgetsize = \Bigoh{\numcolours^{3} \setsize{\pigeonset}}$, 
then we are guaranteed to have enough
vertices~$\clrvertex{\clrvindex}$ to be able to  
process all pre-coloured vertices in this way.

Our final graph
$\graphg = \graphg(\bipartstd)$
is the previous graph
$\alldiffcolgadgets$
with pre-coloured vertices identified with (uncoloured) vertices in
the additional pre-colouring gadget as just described.
Clearly, 
$\graphg$ is \kcolourable if and only if 
the pre-colouring of 
$\alldiffcolgadgets$ can be completed to a legal \kcolouring.
We summarize the properties of our reduction in the 
\ifthenelse{\boolean{conferenceversion}}
{following proposition, stated here without proof.}
{following proposition.}

\begin{proposition}
  \label{pr:graph_construction}
  Given a graph FPHP formula over 
  a left-regular bipartite graph~$\bipartstd$ with left degree
  $d_{\pigeonset} = \numcolours$
  and bounded right degree
  $d_{\holeset} = \bigoh{\numcolours}$, 
  there is an explicit construction of a graph
  $\graphg = \graphg(\bipartstd)$
  such that $\graphg$   has 
  $\Bigoh{\numcolours^{4} \setsize{\pigeonset}}$ 
  vertices of maximal vertex degree~$\bigoh{\numcolours^{2}}$
  and is \kcolourable \ifaoif
  \ifthenelse{\boolean{conferenceversion}}
  {\mbox{Equations \refeq{eq:fphp_defined}--\refeq{eq:fphp_holes}}
    are simultaneously satisfiable.}
  {it is possible to map pigeons to holes in accordance with the
    restrictions in~$\bipartstd$ in a one-to-one fashion,
    \ie \ifaoif
    \mbox{Equations \refeq{eq:fphp_defined}--\refeq{eq:fphp_holes}}
    are simultaneously satisfiable.}
\end{proposition}

\ifthenelse{\boolean{conferenceversion}}
{}
{\begin{proof}
  The number of vertices of 
  $\graphg = \graphg(\bipartstd)$ is at most the number of
  vertices of 
  $\alldiffcolgadgets = \alldiffcolgadgets(\bipartstd)$ 
  plus 
  $\precolgadgetsize = \Bigoh{\numcolours^{3} \setsize{\pigeonset}}$
  additional vertices enforcing the pre-colouring.
  The pre-coloured vertices in the injectivity constraint gadgets 
  get at most $2\numcolours - 1$ new neighbours, so their degree is
  still~$\bigoh{\numcolours}$.
  Hence, the number of vertices and the vertex degree bound in 
  \reflem{lem:completion} remain valid.

  To prove the soundness and completeness of the reduction, note that
  any colouring
  $(\colstd_{1},\ldots,\colstd_{\precolgadgetsize})$
  of 
  $(\clrvertex{1},\ldots,\clrvertex{\precolgadgetsize})$ 
  is completely determined by the colouring
  of the first  $\numcolours$\nobreakdash-clique
  $\set{\clrvertex{1},\ldots,\clrvertex{\numcolours}}$, 
  so that
  $\colstd_{\clrvindex} = \colstd_{\clrvindex'}$ 
  holds whenever
  $\clrvindex \equiv \clrvindex' \pmod{\numcolours}$.
    
  Assume that $\graphstd$ has a \kcolouring~$\chi$.
  Every vertex that was pre-coloured with colour $c\in[\numcolours]$
  in~$\alldiffcolgadgets$ has colour $\vcolour{\clrvertex{c}}$ in $\graphstd$.
  This means that the colouring of $\graphstd$ induces, up to renaming
  of colours,  a completion of the partial colouring
  of~$\alldiffcolgadgets$. 
  By \reflem{lem:completion}, this implies that the FPHP instance is
  satisfiable. 
    
  In the other direction, if the FPHP instance is satisfiable 
  we do as in the proof of \reflem{lem:completion} and colour each
  vertex~$\pigeonindex$ by the colour~$\colstd_i$ corresponding to the
  edge it uses to fly to its hole in some one-to-one mapping.
  We also colour vertices
  $\clrvertex{1}, \ldots, \clrvertex{\numcolours}$
  with colours
  $\vcolour{\clrvertex{\clrvindex}} = \clrvindex$, and extend this
  colouring to the rest of 
  the vertices in the pre-colouring gadget in~\reffig{fig:gadgetseq}.
  By \reflem{lem:completion} this partial colouring can be completed
  to obtain a  legal \kcolouring.
\end{proof}}

Since our reduction encodes local injectivity constraints into local
colouring constraints, it stands to reason that we should be able to
translate between these two types of constraints using low degree
derivations.  In particular, it seems reasonable to expect that any
low-degree refutation of the \kcolouring problem
for~$\graphg(\bipartstd)$ should yield a low-degree refutation for the
functional pigeohole principle on~$\bipartstd$.  This is indeed the
case, as stated in the next lemma.

\begin{lemma}
  \label{lem:reduction}
  Consider the graph
  $\graphstd = \graphstd (\bipartstd)$ 
  obtained from a bipartite graph~$\bipartstd$ as
  in 
  \refpr{pr:graph_construction}.
  If the
  $\numcolours$-colourability constraints
  \mbox{\eqref{eq:colouring_defined}--%
    \eqref{eq:colouring_edges}}
  for $\graphstd$ have a \PC refutation in degree~$\degreestd$,
  then the functional pigeonhole principle constraints 
  \mbox{\refeq{eq:fphp_defined}--\refeq{eq:fphp_holes}}
  defined over~$\bipartstd$ have a PC refutation of degree at
  most~$2\degreestd$. 
\end{lemma}

We will spend what remains of this section on proving this lemma.
The proof is quite similar in spirit to that of
\refpr{pr:alternative-encoding-reduction}.
We start by assuming that we have a PC refutation of
\mbox{Equations \eqref{eq:colouring_defined}--%
  \eqref{eq:colouring_edges}}
in degree $\degreestd$. 
Our first step is to 
substitute all variables~$\varx_{v,j}$ in this refutation with
polynomials of degree at most~$2$ in
variables~$p_{\pigeonindex,\holeindex}$. 
In the second step, we argue that if we apply this substitution to the
axioms in~\mbox{\eqref{eq:colouring_defined}--\eqref{eq:colouring_edges}}, 
then we can derive the resulting substituted polynomials from
\mbox{Equations~\refeq{eq:fphp_defined}--\refeq{eq:fphp_holes}} 
by PC derivations in low degree.
Taken together, this yields a PC refutation in low degree of
the FPHP instance
\mbox{\refeq{eq:fphp_defined}--\refeq{eq:fphp_holes}}.

To describe the substitution, let us focus on a single
gadget~$\diffcolgadgetstd$. The first step is to express all equations
for this gadget as equations over variables
$\varx_{i,1},\ldots \varx_{i,\numcolours}, \varx_{i',1}, \ldots
\varx_{i',\numcolours}$. 
Note that these variables are  essentially the same as those from the
pigeonhole principle instance, except that instead of
$p_{\pigeonindex,\holeindex}$
we use the variable~$\varx_{\pigeonindex,\colstd}$
where $\colstd$ is the number of the edge pigeon~$\pigeonindex$ uses
to fly to hole~$\holeindex$, but for the sake of exposition we want to
keep using the language of colourings.

Let $w$ and $w'$ be the vertices that are supposed to be pre-coloured
with colours $\colstd$ and $\colstd'$, respectively. We stress that now we are
considering the graph $\graphstd$ which has no pre-coloured vertices,
and in particular all the variables mentioning the vertices~$w$ and $w'$ are
unassigned. Recall that  $w$ and~$w'$ also appear in the gadget
depicted in \reffig{fig:gadgetseq}, where they are identified with
some vertices $\clrvertex{\clrvindex}$ and~$\clrvertex{\clrvindex'}$ such that
$\clrvindex \equiv \colstd 
$
and
$\clrvindex' \equiv \colstd' \pmod{\numcolours}$.

For any pair 
$(\colalt,\colalt')$ 
of colours different from $(\colstd, \colstd')$,
Claim~\ref{clm:gadget} guarantees that we can pick some colouring
$\chi_{(\colalt,\colalt')}$ for the gadget
$\diffcolgadgetstd$
such that 
$\chi_{(\colalt,\colalt')}(\pigeonindex) = \colalt$,
$\chi_{(\colalt,\colalt')}(\pigeonindex') = \colalt'$,
$\chi_{(\colalt,\colalt')}(w) = \colstd$ and
$\chi_{(\colalt,\colalt')}(w')=\colstd'$.
Fix for the rest of this proof
such a colouring
$\chi_{(\colalt,\colalt')}$ 
for the gadget~$\diffcolgadgetstd$
for every
$(\colalt,\colalt') \neq (\colstd, \colstd')$.
Then we can write the colour of any vertex~$v$ in
$\diffcolgadgetstd$
other than the pigeon vertices
$\pigeonindex$
and~$\pigeonindex'$ as a  function of 
$(\colalt,\colalt')$. 
In more detail,  we can express every variable
$\varx_{v,\holeindex}$, for $v\not\in \{\pigeonindex,\pigeonindex'\}$,
as a degree\nobreakdash-$2$ polynomial    over the variables 
$\varx_{\pigeonindex,1},\ldots \varx_{\pigeonindex,\numcolours},
\varx_{\pigeonindex',1},\ldots 
\varx_{\pigeonindex',\numcolours}$ 
by
summing over the monomials
$\varx_{\pigeonindex,\colalt}\varx_{\pigeonindex',\colalt'}$
corresponding to the choices of colours
$(\colalt,\colalt')$
for
$(\pigeonindex, \pigeonindex')$
for which the colouring
$\chi_{(\colalt,\colalt')}$
assigns colour~$j$ to vertex~$v$, 
or in symbols
\begin{equation}
  \label{eq:substitution}
  \varx_{v,\holeindex} 
  \mapsto %
  \sum_{
    (\colalt, \colalt') \neq (\colstd, \colstd'), \,
      \chi_{(\colalt, \colalt')}(v) = \holeindex}
  \varx_{\pigeonindex,\colalt}\varx_{\pigeonindex',\colalt'} 
  \eqperiod
\end{equation}
Notice that for the vertices~$w$ and~$w'$ the substitutions we obtain 
from~\refeq{eq:substitution} are
\begin{subequations}
  \begin{align}
    \label{eq:precoloured_substitution-1}
     \varx_{w,\colstd} 
    &
      \substarrowdisplayed
       \sum_{
       {(\colalt, \colalt') \neq (\colstd, \colstd')}
       }
       \varx_{\pigeonindex,\colalt}\varx_{\pigeonindex',\colalt'} 
            \eqcomma
    \\
    \label{eq:precoloured_substitution-2}
    \varx_{w',\colstd'} 
    &
      \substarrowdisplayed
       \sum_{
       { (\colalt, \colalt')\neq (\colstd, \colstd')}
       }
       \varx_{\pigeonindex,\colalt}\varx_{\pigeonindex',\colalt'} 
            \eqcomma
    \\
    \label{eq:precoloured_substitution-3}
    \varx_{w,\colalt} 
    &
      \substarrowdisplayed
            0 && \text{(for $ \colstd \neq \colalt)$,}
    \\
    \label{eq:precoloured_substitution-4}
    \varx_{w',\colalt'} 
    &
      \substarrowdisplayed
            0 && \text{(for $\colstd'\neq \colalt'$),}
  \end{align}
\end{subequations}
since $w$ always gets colour~$\colstd$
and  $w'$ always gets colour~$\colstd'$
in any colouring~$\chi_{(\colalt, \colalt')}$.

\ifthenelse{\boolean{conferenceversion}}
{}
{\begin{example}
  \label{ex:substitution}
  Let us give a concrete example to make clearer how the substitution
  in~\refeq{eq:substitution} works.
  Suppose
  $\numcolours=3$ and consider an arbitrary
  (non-pigeon)
  vertex $v$ in some  gadget
  $\diffcolgadget{\pigeonindex}{\pigeonindex'}{3}{1}$.
  For every pair of colours
  $
  (\colalt, \colalt') \in
  ([3] \times [3]) \setminus \set{(3,1)}
  $
  for
  $(\pigeonindex,\pigeonindex')$
  we have fixed some (legal) colouring of the whole gadget.
  Suppose that in these fixed colourings it holds for some 
  vertex~$v$ in the gadget that
  \begin{enumerate}
  \item \itemsep=0pt
    $v$ takes colour $1$ for 
    $
    (\colalt, \colalt') \in
    \{(1,2),(2,2),(3,2),(3,3)\}$,
    
  \item \itemsep=0pt
    $v$ takes 
    colour $2$ for $     (\colalt, \colalt') \in \{(1,1)\}$ and 
    
  \item \itemsep=0pt
    $v$ takes 
    colour $3$ for $     (\colalt, \colalt') \in 
\{(1,3),(2,1),(2,3)\}$. 
  \end{enumerate}
  In this case, the substitutions for the variables
  $\varx_{v,j}$  mentioning~$v$
  become
  \begin{subequations}
    \begin{align}
      \varx_{v,1} 
      &
        \substarrowdisplayed
        \varx_{\pigeonindex,1}\varx_{\pigeonindex',2} 
        + \varx_{\pigeonindex,2}\varx_{\pigeonindex',2} 
        + \varx_{\pigeonindex,3}\varx_{\pigeonindex',2}
        +  \varx_{\pigeonindex,3}\varx_{\pigeonindex',3}
        \eqcomma
      \\
      \varx_{v,2} 
      &
        \substarrowdisplayed
        \varx_{\pigeonindex,1}\varx_{\pigeonindex',1} 
        \eqcomma
      \\
      \varx_{v,3}
      &
        \substarrowdisplayed
        \varx_{\pigeonindex,1}\varx_{\pigeonindex',3} 
        + \varx_{\pigeonindex,2}\varx_{\pigeonindex',1} 
        + \varx_{\pigeonindex,2}\varx_{\pigeonindex',3}
        \eqperiod
    \end{align}
  \end{subequations}
  We hope that the substitutions in this example can serve as
  a guide for following the rest of the proof.
\end{example}
}

Let us next discuss how the polynomials obtained from
\mbox{\eqref{eq:colouring_defined}--\eqref{eq:colouring_edges}}
after the substitution~\refeq{eq:substitution} 
can be derived in PC from
\mbox{\refeq{eq:fphp_defined}--\refeq{eq:fphp_holes}}.
More precisely we argue that all substituted axioms can be derived
from the equations 
\begin{subequations}
\begin{align}
  \label{eq:axiom_start-1}
    \sum^{\numcolours}_{\colalt=1} \varx_{\pigeonindex,\colalt} &= 1 
    \eqcomma
  \\
  \label{eq:axiom_start-2}
  \varx_{\pigeonindex,\colalt}\varx_{\pigeonindex,\colalt'} &= 0
  && \text{(for $\colalt \neq \colalt'$),}
  \\
  \label{eq:axiom_start-3}
    \sum^{\numcolours}_{\colalt'=1} \varx_{\pigeonindex',\colalt'} &= 1
    \eqcomma
  \\
  \label{eq:axiom_start-4}
  \varx_{\pigeonindex',\colalt}\varx_{\pigeonindex',\colalt'} &= 0
  && \text{(for $\colalt \neq \colalt'$),}
  \\
  \label{eq:axiom_start-5}
    \varx_{\pigeonindex,\colstd}\varx_{\pigeonindex',\colstd'} &= 0 \eqcomma
\end{align}
\end{subequations}
which are just the same, except for variables renaming, as the
pigeon axioms~\refeq{eq:fphp_defined} and~\refeq{eq:fphp_unique} for
pigeons~$\pigeonindex$ and~$\pigeonindex'$ plus the 
collision axiom~\refeq{eq:fphp_holes} for the hole which is the
common neighbour of $\pigeonindex$ and $\pigeonindex'$.
In what follows 
we will need 
the
equation
\begin{equation}
  \label{eq:axiom_more}
  \sum^{\numcolours}_{\colalt=1} \sum^{\numcolours}_{\colalt'=1}
  \varx_{\pigeonindex,\colalt}\varx_{\pigeonindex',\colalt'}
  - \varx_{\pigeonindex,\colstd}\varx_{\pigeonindex',c'} 
  -1 = 0 
\end{equation}
which has the \mbox{degree-$2$} proof 
\begin{equation}
  \sum^{\numcolours}_{\colalt=1} \varx_{\pigeonindex,\colalt} \left( \sum^{\numcolours}_{\colalt'=1}
  \varx_{\pigeonindex',\colalt'} - 1  \right) +
\left( \sum^{\numcolours}_{\colalt=1} \varx_{\pigeonindex,\colalt}  -
  1\right)
  - \varx_{\pigeonindex,\colstd}\varx_{\pigeonindex',c'} 
  = 0 
\end{equation}
from
\mbox{\refeq{eq:axiom_start-1}--\refeq{eq:axiom_start-5}}.

We consider first 
axioms
$\sum^{\numcolours}_{j=1} \varx_{v,j} = 1 $
as in~\refeq{eq:colouring_defined}
for vertices~$v$ 
that are not a pigeon vertex~$\pigeonindex$ or~$\pigeonindex'$. 
It is straightforward to verify that 
such an
axiom after substitution 
as in~\refeq{eq:substitution} 
becomes an equality on the 
\ifthenelse{\boolean{conferenceversion}}
{form~\refeq{eq:axiom_more}.}
{form~\refeq{eq:axiom_more}.%
  \footnote{Consider the example of substitution we saw in 
    \refex{ex:substitution}.
    The
    equation $\varx_{v,1}+\varx_{v,2}+\varx_{v,3}=1$ after substitution becomes
    \begin{equation}
      \left( \varx_{\pigeonindex,1}\varx_{\pigeonindex',2} +
        \varx_{\pigeonindex,2}\varx_{\pigeonindex',2} +
         \varx_{\pigeonindex,3}\varx_{\pigeonindex',2}+ 
         \varx_{\pigeonindex,3}\varx_{\pigeonindex',3} \right) 
       + 
       \left( \varx_{\pigeonindex,1}\varx_{\pigeonindex',1} \right) 
       + 
       \left( 
         \varx_{\pigeonindex,1}\varx_{\pigeonindex',3} + 
         \varx_{\pigeonindex,2}\varx_{\pigeonindex',1} + 
         \varx_{\pigeonindex,2}\varx_{\pigeonindex',3}
       \right) 
       = 1
     \end{equation}
     which is equivalent to
     \begin{equation}
       \left( \varx_{\pigeonindex,1} +
         \varx_{\pigeonindex,2} + \varx_{\pigeonindex,3} \right) 
       \left( \varx_{\pigeonindex',1} + 
         \varx_{\pigeonindex',2} + \varx_{\pigeonindex',3} \right) 
       - \varx_{\pigeonindex,3}\varx_{\pigeonindex',1} = 1\eqperiod
     \end{equation} }}
   If $v$ is a pigeon vertex~$\pigeonindex$ or~$\pigeonindex'$, 
   then no substitution is made and we simply keep the 
   axiom~\eqref{eq:axiom_start-1}
   or~\eqref{eq:axiom_start-3}, respectively.

   Next, we consider axioms~\refeq{eq:colouring_unique}
   on the form $\varx_{v,j} \varx_{v,j'} =  0$,
   where we assume that $v$ is not a pigeon vertex~$\pigeonindex$
   or~$\pigeonindex'$ since in that case we have one of the
   axioms~\eqref{eq:axiom_start-2} and~\eqref{eq:axiom_start-4}. 
   After substitution
   an axiom~\eqref{eq:colouring_unique}
   for $v \notin \set{i,i'}$
   becomes a sum of \mbox{degree-$4$} terms of the form
   $
   \varx_{\pigeonindex,\colalt_{1}}
   \varx_{\pigeonindex',\colalt'_{1}}
   \varx_{\pigeonindex,\colalt_{2}}
   \varx_{\pigeonindex',\colalt'_{2}}
   $.
   Recall that the substitution associates disjoint sets of pairs
   $(\colalt,\colalt')$ to the colours 
   \ifthenelse{\boolean{conferenceversion}}
   {for~$v$.}
   {for~$v$ (see again Example~\ref{ex:substitution}).} 
   Therefore, for each term
   $
   \varx_{\pigeonindex,\colalt_{1}}
   \varx_{\pigeonindex',\colalt'_{1}}
   \varx_{\pigeonindex,\colalt_{2}}
   \varx_{\pigeonindex',\colalt'_{2}}
   $
   it must be that either $\colalt_{1}\neq\colalt_{2}$ or
   $\colalt'_{1}\neq\colalt'_{2}$ 
   holds,  %
   and such a term can be
   derived from~\eqref{eq:axiom_start-2} or~\eqref{eq:axiom_start-4}
   by multiplication.

   Let us finally consider axioms on the form
   $\varx_{u,j}\varx_{v,j} = 0$
   for
   $(u,v) \in E(\graphstd)$
   as in~\eqref{eq:colouring_edges}. 
   There is no edge between $\pigeonindex$ and $\pigeonindex'$ in our
   constructed graph, 
   so for the size of the intersection
   between $\set{u,v}$ and $\set{\pigeonindex,\pigeonindex'}$ 
   it holds that
   $
   0 \leq
   \Setsize{
     \set{u,v} \intersection \set{\pigeonindex,\pigeonindex'}
   }
   \leq 1
   $. 

   If
   $
   \Setsize{
     \set{u,v} \intersection \set{\pigeonindex,\pigeonindex'}
   }
   = 0$,
   then  after substitution the axiom~\eqref{eq:colouring_edges} 
   becomes a sum of degree\nobreakdash-$4$ terms of the form
   $
   \varx_{\pigeonindex,\colalt_{1}}
   \varx_{\pigeonindex',\colalt'_{1}}
   \varx_{\pigeonindex,\colalt_{2}}
   \varx_{\pigeonindex',\colalt'_{2}}
   $.
   Consider any such term. If either $\colalt_{1}\neq\colalt_{2}$ or
   $\colalt'_{1}\neq\colalt'_{2}$, then the term can be derived
   from~\eqref{eq:axiom_start-2} or~\eqref{eq:axiom_start-4}.
   We claim that no term can have
   $\colalt_{1}=\colalt_{2}$ and $\colalt'_{1}=\colalt'_{2}$.
   To see this, note that this would imply that when performing
   substitution as in~\refeq{eq:substitution} the variables
   $\varx_{u,j}$ and~$\varx_{v,j}$ both get expanded to a sum containing
   $\varx_{\pigeonindex,\colalt_{1}}\varx_{\pigeonindex',\colalt'_{1}}$. 
   But this would in turn mean that the colouring
   $\chi_{(\colalt_{1},\colalt_{1}')}$
   that we fixed for the gadget~$\diffcolgadgetstd$ at the start of
   the proof  assigned  colours
   $\chi_{(\colalt_{1},\colalt_{1}')}(u) =
   \chi_{(\colalt_{1},\colalt_{1}')}(v)$, 
   which is impossible since there is an edge between $u$ and $v$ and
   $\chi_{(\colalt_{1},\colalt_{1}')}$ was chosen to be a legal colouring.

   The remaining case is when we have intersection size
   $
   \Setsize{
     \set{u,v} \intersection \set{\pigeonindex,\pigeonindex'}
   }
   = 1$.   
   Without loss of generality because of symmetry we can assume that
   we have an axiom 
   $\varx_{u,j}\varx_{v,j} = 0$
   for 
   $u \notin \set{i,i'}$
   and
   $v=\pigeonindex$.
   The axiom becomes after substitution a sum of terms of the form
   $\varx_{\pigeonindex,\colalt}
   \varx_{\pigeonindex',\colalt'}
   \varx_{\pigeonindex,\holeindex}$.
   If for some term we would have $\colalt=\holeindex$, then
   $\chi_{(\holeindex,\colalt')}$ would assign the same colour
   $\holeindex$ to both $u$ and $\pigeonindex$. This is again
   impossible since $\chi_{(\holeindex,\colalt')}$ is a legal colouring of the
   gadget by construction. Hence we have $\colalt\neq j$ and it
   follows that
   $
   \varx_{\pigeonindex,\colalt}
   \varx_{\pigeonindex',\colalt'}
   \varx_{\pigeonindex,\holeindex}
   $
   is 
   derivable
   from~\eqref{eq:axiom_start-2}.

\jncomment{Actually, would it be possible to get a slightly simpler
  reduction below by applying this restriction to every single variable 
  $\clrvertex{\clrvindex}$
  in the pre-colouring gadget? Or this doesn't work? It doesn't seem to
  important and I don't have time to think about it carefully now, but
  I am just asking\ldots}
\mlcomment{I recall that if we try to do that, then we would have more
cases to analyze (e.g.\ edges incident to a precoloured vertex) while
in this way we do not need it. Anyway your restriction would work as
well but we would need to change the proof.}

We are now almost done with the proof of
\reflem{lem:reduction}.
We have defined how to  substitute  variables~$\varx_{v,j}$ in
\mbox{\eqref{eq:colouring_defined}--\eqref{eq:colouring_edges}}
and have shown that the equations that we obtain after these
substitutions can be derived from
\mbox{Equations~\refeq{eq:fphp_defined}--\refeq{eq:fphp_holes}} 
in low degree.
The final issue that remains it to get rid of all vertices~$\clrvertex{\clrvindex}$
in the pre-colouring gadget in \reffig{fig:gadgetseq} that are not
members of any injectivity constraint gadget
$\diffcolgadgetstd$.
For such variables the substitution is simply an assignment:
we let
$\varx_{\clrvertex{\clrvindex},\colalt} \mapsto 1$ 
when
$\clrvindex \equiv \colalt \pmod{\numcolours}$
and 
$\varx_{\clrvertex{\clrvindex},\colalt} \mapsto 0$ 
otherwise.%
\footnote{Note that 
  here the substitution for $\varx_{\clrvertex{\clrvindex},\colalt}$ where
  $\clrvindex \equiv \colalt \pmod{\numcolours}$
  is different from the one used for
  vertices that are members of some gadget
  $\diffcolgadgetstd$
  in~\refeq{eq:precoloured_substitution-1}
  and~\refeq{eq:precoloured_substitution-2}.
  For variables
  $\varx_{\clrvertex{\clrvindex},\colalt}$ where
  $\clrvindex \not\equiv \colalt\pmod{\numcolours}$ the substitution is the same
  as
  in~\refeq{eq:precoloured_substitution-3}
  and~\refeq{eq:precoloured_substitution-4}, though.}
This immediately satisfies all
axioms~\eqref{eq:colouring_defined}
and~\eqref{eq:colouring_unique}
for these vertices, removing these axioms from the refutation.
It remains to check the axioms~\eqref{eq:colouring_edges} for any
pair of connected vertices~$\clrvertex{\clrvindex}$ and~$\clrvertex{\clrvindex'}$. But by
construction, if $\clrvertex{\clrvindex}$ and~$\clrvertex{\clrvindex'}$ are connected it holds
that $\clrvindex \not\equiv \clrvindex'\pmod{\numcolours}$.
Therefore, for every $\colalt \in [\numcolours]$
we have that either $\varx_{\clrvertex{\clrvindex},\colalt} \mapsto 0$
or $\varx_{\clrvertex{\clrvindex'},\colalt} \mapsto 0$ holds,
regardless of whether these two
vertices are in some gadget
$\diffcolgadgetstd$
or not.

To summarize what we have done, we started with any arbitrary refutation of
\mbox{\eqref{eq:colouring_defined}--\eqref{eq:colouring_edges}}
and substituted all variables with \mbox{degree-$2$} polynomials over
the variables $\varx_{\pigeonindex,\holeindex}$ for $\pigeonindex \in [n]$. 
Then we proved that all these substituted axioms (and therefore the
whole refutation) follow from
Equations~\mbox{\eqref{eq:axiom_start-1}--\eqref{eq:axiom_start-3}}. 
It is straightforward to verify that, up to variable renaming, these
axioms are nothing other than the FPHP axioms in
  \mbox{\refeq{eq:fphp_defined}--\refeq{eq:fphp_holes}}.
This concludes the proof of
\reflem{lem:reduction}.
Putting everything together, we can now state and prove our main theorem.

\begin{theorem}\label{thm:worstcase}
  For any integer $\numcolours \geq 3$ there is an efficiently
  constructible family of graphs $\{\graphstd_{n}\}_{n\in \N}$ with
  $\bigoh{\numcolours^{4}n}$ 
  vertices of degree $\bigoh{k^{2}}$ that do not possess
  \kcolouring{}s, but for which the corresponding system of polynomial
  equations 
  \mbox{\eqref{eq:colouring_defined}--\eqref{eq:colouring_edges}}
  require  degree~$\bigomega{n}$, and hence size~$\exp(\bigomega{n})$,
  to be refuted in polynomial calculus.
\end{theorem}

\begin{proof}
  We need to build a family of bipartite graphs
  $\{\bipartstd_{n}\}_{n\in \N}$ with the properties needed to apply
  \refth{thm:fphp} and \refpr{pr:graph_construction}. This would yield
  a family of graphs $\{\graphstd_{n}\}_{n\in \N}$ as in the
  theorem statement.
  Indeed any sublinear degree refutation for $\numcolours$-colouring of
  $\graphstd_{n}$ would imply, by \reflem{lem:reduction}, a sublinear
  degree refutation for the functional pigeonhole principle for
  $\bipartstd_{n}$, but this is impossible by a construction
  of~$\bipartstd_{n}$ that triggers the lower bound in \refth{thm:fphp}.

  For every $\numcolors \geq 3$ the starting point of our construction
  is the paper by Alon and Capalbo~\cite{AlonCapalbo2002Explicit},
  which provide universal constants $\alpha, \delta$ and simple graphs
  $\{H_{n}\}_{n\in\N}$ of degree at most $\numcolors$ so that $H_{n}$
  has $n$ vertices and for every set $X$ of vertices in $H_{n}$, the
  boundary of $X$ in $V(H_{n})\setminus X$ has size at least
  $\delta|X|$.
  Consider the bipartite graph obtained taking two copies of the
  vertices of $H_{n}$, one identified with the left side and one with
  the right side, and adding an edge from the left copy of $u$ to the
  right copy of $v$ for every edge $\{u,v\} \in E(H_{n})$. Then take
  an arbitrary vertex $\hat{u} \in V(H_{n})$ and remove the
  corresponding right copy from the graph.
  The resulting graph, that we define to be $\bipartstd_{n}$, has $n$
  vertices on the left side, $n-1$ vertices on the right side and
  degree at most $\numcolors$ on both sides. Moreover for every set
  $I$ of left vertices of size at most $\alpha n$ the boundary of $I$
  consists in the right copies of the corresponding boundary in
  $H_{n}$, minus vertex $\hat{u}$. Therefore $\bipartstd_{n}$
  satisfies the properties needed to apply \refth{thm:fphp}.
\end{proof}

\ifthenelse{\boolean{conferenceversion}}
{\section{Short Proofs
  for $k$-Colouring Instances
  in Cutting Planes}}
{\section{Polynomial-Length Proofs
  for $k$-Colouring Instances
  in Cutting Planes}}
\label{sec:cp-refutation}

Theorem~\ref{thm:worstcase} tells us that there are
non-\kcolourable graphs~$\graphstd_{n}$ for which it is impossible for
polynomial calculus to certify non-\kcolourability efficiently.
As is clear from our reduction, the \kcolouring formulas for these
graphs  are essentially obfuscated instances of the functional
pigeonhole principle. 

It is well-known that cutting planes can easily prove that pigeonhole
principle formulas are unsatisfiable by just counting the number of
pigeons and holes and deduce that the pigeons are too many to fit in
the holes~\cite{CCT87ComplexityCP}. 
As we show in this section, the instances of
$\numcolours$-colouring obtained via the reduction from FPHP also have
short cutting planes refutations.  What these refutations do is
essentially to ``de-obfuscate'' the \kcolouring formulas to recover
the original functional pigeonhole principle instances, which can then
be efficiently refuted.

We are going to describe our cutting planes refutation as a decision
tree such that at every leaf we have a cutting planes
refutation of the formula restricted by the partial assignment defined
by the tree branch reaching that leaf. These refutations of the
restricted versions of the formula can then be combined to yield a
refutation of the original, unrestricted formula as stated in 
\reflem{lem:cp_weakening}
and
\refpr{pr:refutation-brute-force}.
\ifthenelse{\boolean{conferenceversion}}
{The proofs of these statements are fairly routine and we omit them in
  this conference version of the paper.}
{The proofs of these statements are fairly routine, but we include them
  below for completeness.}

We recall that as discussed in \refsec{sec:preliminaries} we will use
$  \sum_{i} \lincoeff[i]\varx_{i}  \leq \linconst $ 
as an alias for
\mbox{$  \sum_{i} - \lincoeff[i]\varx_{i}  \geq -\linconst $}
and 
$  \sum_{i} \lincoeff[i]\varx_{i}  = \linconst $ 
as an alias for the combination of
$  \sum_{i} \lincoeff[i]\varx_{i}  \leq \linconst $ 
and
$  \sum_{i} \lincoeff[i]\varx_{i}  \geq \linconst $.
In particular, we will frequently write $\varx=b$ for some
variable~$\varx$ and $b \in \set{0,1}$ as a shorthand for the pair of
inequalities $\varx \leq b$ and $-\varx \leq -b$.

\begin{lemma}
  \label{lem:cp_weakening}
  Let $b \in \{0,1\}$ and
  suppose that there exists a cutting planes derivation 
  $
  (\linaux[1], \ldots, \linaux[\cplength])
  $
  in length~$\cplength$
  of the inequality 
  $
  \sum_{i} \lincoeff[i]\varx_{i} \leq \linconst
  $ 
  from the system of inequalities $\linset{S} \cup \{\varx=b\}$.  
  Then for some $K \in \N$ there is a CP derivation in length
  $\bigoh{\cplength}$
  of the  inequality
  \begin{equation}
    \label{eq:desired_inequality}
    (-1)^{1-b} K \cdot (\varx -b) + \sum_{i}
    \lincoeff[i]\varx_{i} \leq \linconst 
  \end{equation}
  from~$\linset{S}$.
\end{lemma}

\ifthenelse{\boolean{conferenceversion}}
{}
{
\begin{proof}
  We only consider the case of $b=1$, since the case of $b=0$ is
  essentially the same.
  Observe that $\varx \leq 1$ is a boolean axiom, 
  and so
  the only axiom lost when passing from $\linset{S} \cup
  \{\varx=b\}$ to $\linset{S}$ is $ - \varx \leq - 1 $.
  The proof is by forward induction over the derivation. 
  We use the notation
  $
  \linaux[t] 
  = \sum_{i} \lincoeff[i]^{(t)}\varx_{i} \leq
  \linconst^{(t)}
  $
  for the linear inequalities
  $\linaux[t]$,  $t \in [\cplength]$,
  in the original derivation.
  We show how to derive 
  the inequalities
  \begin{equation}
    \label{eq:cp_weakening_generic}
    K^{(t)} \cdot (\varx-1) + \sum_{i} \lincoeff[i]^{(t)}\varx_{i}
  \leq \linconst^{(t)}
  \end{equation}
  for $t \in [\cplength]$
  with a constant number of rule
  applications per step, assuming that the 
  preceding inequalities have already been derived.
  We proceed by cases depending on which rule
  was
  used to derive $\linaux[t]$ in the original derivation.
  
  If $\linaux[t]$ is the axiom $ -\varx \leq -1 $ we can substitute it
  with axiom $ 0 \leq 0$, which can be
  written 
  as
  $ (\varx - 1 ) - \varx \leq -1$. 
  Notice that, technically speaking, $0 \leq 0$ 
  is not an axiom of the cutting planes proof system, 
  but it is convenient to allows such lines in our derivation and
  remove them in a final postprocessing step.
  If $\linaux[t]$ is a variable axiom $\varx \geq 0$ or an initial
  axiom $\lin[j] \in \linset{S}$, then it is already on the
  form~\eqref{eq:cp_weakening_generic} with $K^{(t)}=0$.

  If $\linaux[t]$ is derived as a sum
  $\linaux[t] = \linaux[t_{1}]+\linaux[t_{2}]$ 
  for some
  $1\leq t_{1} < t_{2} < t$,
  then by 
  the
  induction hypothesis we have already derived 
  \begin{equation}
    K^{(t_{1})}(\varx-1) + \sum_{i} \lincoeff[i]^{(t_{1})} \varx_{i}
    \leq \linconst^{(t_{1})}
  \end{equation}
  and
  \begin{equation}
    K^{(t_{2})}(\varx-1) + \sum_{i} \lincoeff[i]^{(t_{2})} \varx_{i}
    \leq \linconst^{(t_{2})}
  \end{equation}
  and the sum of these two inequalities is already on the
  form~\eqref{eq:cp_weakening_generic} with
  $K^{(t)} = K^{(t_{1})} + K^{(t_{2})}$.
  
  If $\linaux[t]$ is derived 
  by multiplication
  $\linaux[t] = \alpha\linaux[t']$ 
  for some $1\leq t' < t$, then
  then by 
  the
  induction hypothesis we have already derived 
  \begin{equation}
    K^{(t')}(\varx-1) + \sum_{i} \lincoeff[i]^{(t')} \varx_{i} 
    \leq \linconst^{(t')}
    \eqcomma
  \end{equation}
  and so 
  $\linaux[t] = \alpha\linaux[t']$ 
  is on the
  form~\eqref{eq:cp_weakening_generic} with
  $K^{(t)} = \alpha K^{(t')}$.
  
  If $\linaux[t]$ is obtained by the application of the division rule
  to some previously derived inequality~$\linaux[t']$, then
  $\linaux[t]$ 
  has the form
  \begin{equation}
    \label{eq:legal_div}
    \sum_{i} \frac{\lincoeff[i]^{(t')}}{\lindiv}\varx_{i}  \leq
    \left\lfloor
      \frac{\linconst^{(t')}}{\lindiv} 
    \right\rfloor
    \eqperiod
  \end{equation}
  By
  the
  induction hypothesis
  we have already derived
  \begin{equation}
    \label{eq:almost_div}
    K^{(t')}(\varx - 1) + \sum_{i} \lincoeff[i]^{(t')}\varx_{i}  
    \leq \linconst^{(t')}
  \end{equation}
  and we want to divide this inequality by~$\lindiv$. In order to do
  so, however, we need to ensure that all coefficients
  in~\refeq{eq:almost_div} are divisible by~$\lindiv$.
  Since by assumption the  application of the
  division rule  in~\eqref{eq:legal_div} was legal 
  we have that
  $\lindiv$ divides all of $\lincoeff[1], \ldots, \lincoeff[n]$.
  Choose the smallest~$K'$ divisible by~$\lindiv$ such that
  $K' - K^{t'} \geq 0$,
  multiply
  $\varx \leq 1$,
  which we can also write as 
  $\varx - 1 \leq 0$, 
  by
  $K' - K^{t'}$
  and then add to~\eqref{eq:almost_div} to obtain
  \begin{equation}
    \label{eq:ready_for_div}
    K'(\varx - 1) + \sum_{i} \lincoeff[i]^{(t')}\varx_{i}  
    \leq \linconst^{(t')}
    \eqperiod
  \end{equation}
  In order to apply division we have to collect all constant terms on
  the right-hand side,
  meaning that we rewrite
  \refeq{eq:ready_for_div}
  as
  \begin{equation}
    K'    \varx  + \sum_{i} \lincoeff[i]^{(t')}\varx_{i}  
    \leq \linconst^{(t')} +     K'
    \eqcomma
  \end{equation}
  and division now yields
  \begin{equation}
    \label{eq:final_div}
    \frac{K'}{\lindiv}\varx 
    + 
    \sum_{i} \frac{\lincoeff[i]^{(t')}}{\lindiv}\varx_{i}  
    \leq 
    \left\lfloor \frac{\linconst^{(t')} + K'}{\lindiv} \right\rfloor
  \end{equation}
  which we can write as
  \begin{equation}
    \label{eq:final_div_right_form}
    \frac{K'}{\lindiv} ( \varx - 1)
    + 
    \sum_{i} \frac{\lincoeff[i]^{(t')}}{\lindiv}\varx_{i}  
    \leq 
    \left\lfloor \frac{\linconst^{(t')} + K'}{\lindiv} \right\rfloor
    = 
    \left\lfloor 
      \frac{\linconst^{(t')}}{\lindiv}
    \right\rfloor 
  \end{equation}
  since $\lindiv$ divides~$K'$.
  The inequality~\refeq{eq:final_div_right_form} is on the
  form~\eqref{eq:cp_weakening_generic}, and so we are done with our
  analysis of the division step.

  By the induction principle we
  obtain a derivation in lenght~$\bigoh{\cplength}$
  of~\refeq{eq:desired_inequality}, as claimed. As the last step we
  remove all occurrences of lines $0 \leq 0$. It is clear that any
  addition of such an inequality to another inequality can simply be ignored.
\end{proof}

We next show how \reflem{lem:cp_weakening} can be used to piece
together refutations of restricted versions of a \kcolouring formula
to one refutation of the unrestricted formula.
}

\begin{proposition}
  \label{pr:refutation-brute-force}
  Let $\graphstd$ be a graph and
  $\numcolours \geq 2$ be a positive integer, 
  and 
  let $\linset{S}$ be the set of inequalities
  \mbox{\refeq{eq:colouring_defined_cp}--\refeq{eq:colouring_edges_cp}}
  for $\graphstd$ and~$\numcolours$.
  If for a fixed set of vertices
  $u_{1}, u_{2}, \ldots, u_{\ell}$ in~$\graphstd$
  and   every choice of colours
  $(\cpcol_{1}, \cpcol_{2}, \ldots, \cpcol_{\ell}) \in [\numcolours]^{\ell}$
  for these vertices there is a CP refutation in length 
  at most~$\cplength$ 
  of
  the set of inequalities
  $
  \linset{S} \cup 
  \set{\varx_{u_{1},\cpcol_{1}}=1, \varx_{u_{2},\cpcol_{2}}=1,
    \ldots, \varx_{u_{\ell},\cpcol_{\ell}}=1}
  $, 
  then there is a CP refutation of~$\linset{S}$ 
  in \mbox{length $\numcolours^{\bigoh{\ell}} \cdot \cplength$}.
\end{proposition}

\ifthenelse{\boolean{conferenceversion}}
{}
{

\begin{proof}
  We prove the claim by induction on $\ell$. If $\ell=0$ then the
  statement is vacuous.
  For 
  $\ell > 0$
  we assume that we can derive $1 \leq 0$ in
  length
  at most $\cplength$ from
  $
  \linset{S} \cup 
  \set{\varx_{u_{1},\cpcol_{1}}=1, \varx_{u_{2},\cpcol_{2}}=1,
    \ldots, \varx_{u_{\ell},\cpcol_{\ell}}=1,  \varx_{u,\cpcol}=1}
  $
  for a fixed vertex $u$ and every $\cpcol \in [k]$.

  For each  $\cpcol \in [\numcolours]$ we use
  \reflem{lem:cp_weakening} to construct a 
  CP derivation 
  of the inequality
  $K_{\cpcol}(\varx_{u,\cpcol} - 1) + 1 \leq 0$ 
  from
  $ 
  \linset{S} \union
  \set{\varx_{u_{1},\cpcol_{1}}=1, \varx_{u_{2},\cpcol_{2}}=1,
    \ldots, \varx_{u_{\ell},\cpcol_{\ell}}=1} 
  $
  in
  length~$\cplengthconst\cplength$, where $\cplengthconst$ is a
  universal constant.
  By dividing each such inequality by 
  $K_{\cpcol}$ we get $\varx_{u,\cpcol} \leq 0$ 
  for all $\cpcol \in [\numcolors]$.
  By summing all 
  these inequalities with the initial axiom 
  $ \sum_{\cpcol} \varx_{u,\cpcol} \geq 1$ 
  we 
  obtain  
  $0 \geq 1$.
  The total length of 
  this
  refutation is
  $\numcolours \cplengthconst \cplength + 
  2 \numcolours
  $.
  The proposition 
  follows by the 
  induction principle.
\end{proof}
}

We can now state the main result of this section, namely that the hard
\kcolouring instances for polynomial calculus constructed in
\refsec{sec:pc-lower-bound} are easy for cutting planes.

\begin{proposition}
  Let $\bipartstd$ be a left-regular bipartite graph~$\bipartstd$ with
  left degree~$ \numcolours$
  and bounded right degree~$\bigoh{\numcolours}$, 
  and  consider 
  the graph $\graphstd=\graphstd(B)$ in 
  \refpr{pr:graph_construction}.
  Then if there is no complete matching of the left-hand side
  of~$\bipartstd$ into the right-hand side, then 
  the set of inequalities
  \mbox{\refeq{eq:colouring_defined_cp}--\refeq{eq:colouring_edges_cp}}
  encoding the
  $\numcolours$-colouring problem on~$\graphstd$ has
  a cutting planes refutation 
  in length
  $
  {\numcolours}^{\bigoh{\numcolours}} \cdot 
  \setsize{V(\bipartstd)}^{\bigoh{1}}
  $.
\end{proposition}

\begin{proof}[Proof sketch]
  Consider the first $\numcolours$~vertices
  $\clrvertex{1}, \ldots, \clrvertex{\numcolours}$ 
  in the pre-colouring gadget in~$\graphstd$
  as depicted in
  \reffig{fig:gadgetseq}, which form a \mbox{$\numcolours$-clique}.
  For every partial colouring
  $(\cpcol_{1}, \cpcol_{2},\ldots, \cpcol_{\numcolours}) \in
  [\numcolours]^{\numcolours}$
  of this  \mbox{$\numcolours$-clique}
  we build a cutting planes refutation of 
  \begin{equation}
    \label{eq:colouring-instance-to-be-refuted}
    \linset{S} \cup 
    \set{\varx_{\clrvertex{1},\cpcol_{1}}=1, \varx_{\clrvertex{2},\cpcol_{2}}=1,
      \ldots, \varx_{\clrvertex{\numcolours},\cpcol_{\numcolours}}=1 } 
    \eqperiod
  \end{equation}
  The result then follows by combining all of these refutations using
  \refpr{pr:refutation-brute-force}.

  Fix a choice of colours
  $(\cpcol_{1}, \cpcol_{2},\ldots, \cpcol_{\numcolours}) \in
  [\numcolours]^{\numcolours}$.
  Notice that if some colour occurs twice in this tuple, then we
  can derive contradiction in 
  length~$\bigoh{1}$ 
  from~\refeq{eq:colouring-instance-to-be-refuted}
  since one of the edge axioms~\refeq{eq:colouring_edges_cp} is violated.
  Suppose therefore that
  $(\cpcol_{1}, \cpcol_{2},\ldots, \cpcol_{\numcolours})$ 
  is a permutation of
  $[\numcolours]$. We will construct a CP refutation  
  of~\refeq{eq:colouring-instance-to-be-refuted}
  in length
  $
  {\numcolours}^{\bigoh{\numcolours}} \cdot 
  \setsize{V(\bipartstd)}^{\bigoh{1}}
  $.

  The system of inequalities $\linset{S}$ is symmetric which respect
  to the permutation of the colour indices,
  so \wolog we focus on giving a refutation for 
  \begin{equation}
    \label{eq:colouring-instance-symmetry}
    \linset{S} \cup 
    \set{\varx_{\clrvertex{1},1}=1, \varx_{\clrvertex{2},2}=1, \ldots,
      \varx_{\clrvertex{\numcolours},\numcolours}=1 } 
    \eqperiod
  \end{equation}
  The equations
  $
  \set{\varx_{\clrvertex{1},1}=1, 
    \varx_{\clrvertex{2},2}=1, 
    \ldots,
    \varx_{\clrvertex{\numcolours},\numcolours}=1 
  }
  $ 
  taken together with~$\linset{S}$  allow us to efficiently  infer 
  $\varx_{\clrvertex{i},i \bmod{\numcolours}}=1$ 
  for all the vertices $\clrvertex{i}$, $i \in [\precolgadgetsize]$,
  in the gadget in \reffig{fig:gadgetseq}
  (where we recall from 
  \refsec{sec:preliminaries}
  that we identify colours~$0$ and~$\numcolours$ when convenient).
  The resulting set of equalities and inequalities
  $
  \linset{S} \union 
  \setdescr{\varx_{\clrvertex{i},i \bmod{\numcolours}}=1}
  {i \in [\precolgadgetsize]}
  $
  is essentially an encoding of the $k$-colouring problem for the
  partially colored graph $\alldiffcolgadgets$ 
  in \reflem{lem:completion}
  consisting of the gadgets in
  \reffig{fig:gadget}.
  Indeed, since the partial assignment 
  $\set{\varx_{\clrvertex{1},1}=1, \varx_{\clrvertex{2},2}=1, \ldots,
    \varx_{\clrvertex{\numcolours},\numcolours}=1 }$ 
  forces the colours of all vertices
  $\clrvertex{i}$, $i \in [\precolgadgetsize]$,
  in \reffig{fig:gadgetseq},
  this gives us   back the pre-coloured vertices in the gadgets in
  \reffig{fig:gadget}.

  As argued in (the proof of)  \reflem{lem:completion},
  $\alldiffcolgadgets$ is the union of at most 
  $\Bigoh{\numcolours^{2} \setsize{V(\bipartstd)}}$
  injectivity constraint gadgets 
  $\diffcolgadget{i}{i'}{\cpcol}{\cpcol'}$ 
  that forbid
  pigeons~$i$ and~$i'$ taking their $\cpcol$th and $\cpcol'$th
  edges, respectively,  colliding in some hole~$j$.
  If we introduce the alias
  $p_{i,j}$
  for
  $x_{i,\cpcol}$,
  where $j$~is the hole to which the $\cpcol$th edge from pigeon~$i$
  leads, then our goal can be described as deriving the 
  pigeonhole axiom
  $
  p_{i,j} +   p_{i',j} 
  = 
  \varx_{i,\cpcol} + \varx_{i',\cpcol'} 
  \leq 
  1
  $
  from the set of inequalities of
  the corresponding gadget 
  $\diffcolgadget{u}{v}{\cpcol}{\cpcol'}$.
  We will see shortly how to do so in length
  $\Bigoh{\numcolours^{\bigoh{\numcolours}}}$.
  Once we extract these pigeonhole inequalities we observe that the
  collection of these inequalities together with the
  inequalities~\eqref{eq:colouring_defined_cp} form a cutting plane
  encoding 
  \begin{subequations}
    \begin{align}
      \label{eq:php_cp_pigeon}
      \sum_{\holeindex \in \holesforpigeon{\pigeonindex}}
      p_{\pigeonindex,\holeindex} 
      &\geq 1
      &&
         \text{$\pigeonindex \in \pigeonset$,}
      \\
      \label{eq:php_cp_hole}
      p_{\pigeonindex ,\holeindex} + p_{\pigeonindex ',\holeindex} 
      &\leq 0
      &&
         \text{$\pigeonindex \neq  \pigeonindex ' \in \pigeonset$,
         $\holeindex \in \holesforpigeon{\pigeonindex} \intersection
         \holesforpigeon{\pigeonindex'}$.}
    \end{align}
  \end{subequations}
  of the graph pigeonhole principle on the bipartite
  graph~$\bipartstd$ with left-hand side~$\pigeonset$ and right-hand
  side~$\holeset$.  
  Such a system of inequalities has
  a cutting plane refutation in length
  $\Bigoh{{\setsize{V(\bipartstd)}}^{3}}$~\cite{CCT87ComplexityCP}.

  In order to derive 
  $\varx_{i,\cpcol} + \varx_{i',\cpcol'} \leq 1$ 
  we consider the inequalities involving vertices of
  $\diffcolgadget{i}{i'}{\cpcol}{\cpcol'}$
  plus the equations $\varx_{i,\cpcol}=1$ and $\varx_{i,\cpcol'}=1$. 
  By Claim~\ref{clm:gadget}
  this is an unsatisfiable system of inequalities of size
  $\bigoh{\numcolors}$.
  By the refutational completeness of cutting planes, and
  using \reflem{lem:cp_weakening} twice, we obtain a derivation of
  $ K_{1}(\varx_{i,\cpcol} -1) + K_{2}(\varx_{i',\cpcol'} - 1 )\leq -1$
  in length~$\exp(\bigoh{\numcolors})$.
  Adding multiples of axioms on the form
  $\varx - 1 \leq 0$
  we get 
  the
  inequality
  $ K(\varx_{i,\cpcol} -1) + K(\varx_{i',\cpcol'} - 1 ) \leq -1$ 
  for some positive integer~$K$, 
  and division by~$K$ yields $\varx_{i,\cpcol} + \varx_{i',\cpcol'} \leq 1$.

  We have shown
  how to derive contradiction is length
  ${\numcolours}^{\bigoh{\numcolours}} |V(\bipartstd)|^{\bigoh{1}}$ 
  for any  given colouring of the vertices 
  $\clrvertex{1}, \ldots, \clrvertex{\numcolours}$. 
  We take such refutations for all 
  ${\numcolours^{\numcolours}}$~possible ways of assigning
  colours to these vertices and joint them together using
  \refpr{pr:refutation-brute-force}
  into a refutation of the original, unrestricted formula. The
  proposition follows.
\end{proof}

\section{Concluding Remarks}
\label{sec:conclusion}

In this work we exhibit explicitly constructible graphs which are
non-$\numcolours$-colourable but which require large degree in
polynomial calculus to certify this fact for the canonical encoding of
the \kcolouring problem into polynomial equations over
$\set{0,1}$\nobreakdash-valued variables.  This, in turn, implies that
the size of any polynomial calculus proof of
non-$\numcolours$-colourability for these graphs must be exponential
measured in the number of vertices.
    
Our degree lower bound also applies to a slightly different encoding
with primitive $\numcolours$th roots of unity used
in~\cite{DLMM08Hilbert,DLMM11ComputingInfeasibility} to build
\kcolouring algorithms based on Hilbert's Nullstellensatz. These
algorithms construct certificates of non-$\numcolours$-colourability by 
solving linear systems of equations over the coefficients of all
monomials up to a certain degree. 
The current paper yields explicit instances for which this method needs to
consider monomials up to a very large degree, and therefore has to
produce a linear system of exponential size. This answers an open
question raised in, 
for example,
\cite{DLMM08Hilbert,DLMO09ExpressingCombinatorial,DLMM11ComputingInfeasibility,LLO16LowDegreeColorability}.

This leads to an important observation, however. The degree lower
bound applies to both polynomial encodings discussed above, but the
size lower bound only applies to the encoding using
$\set{0,1}$\nobreakdash-valued variables.  It is still conceivable
that proofs of non-$\numcolours$-colourability in the roots-of-unity
encoding can be small although they must have large degree. This
raises the following question.

\begin{openproblem}
  Is there a family of non-$3$-colourable graphs such that any
  polynomial calculus proof of non-$3$-colourability using the roots
  of unity encoding must require large size?
\end{openproblem}

If the answer to the question is positive, then no matter how we
choose the monomials to consider for the linear system construction
in~\cite{DLMM08Hilbert,DLMM11ComputingInfeasibility}, the size of the
system will have to be large.

To further reduce the size of the linear system, the algorithms
in~\cite{DLMM08Hilbert,DLMM11ComputingInfeasibility} make use of the
symmetries in the graphs.  It is a natural question how much such an
approach could help for our non-$\numcolours$-colourable instances.
It seems plausible that if we apply our construction to a randomly
generated bipartite graph with appropriate parameters, then the final
graph will not have many symmetries except for the local symmetries
inside the gadgets. In that case our lower bound might apply for the
improved version of the algorithm as well.

The work in~\cite{AtseriasOchremiak2018ProofComplexity} addresses the
proof complexity of refuting constraint statisfaction problems (CSP)
and show that standard reductions between CSPs preserve hardness, to
some extent, and in particular degree. These reductions are able to
translate between various encodings, and indeed they may not preserve
monomial size.

\jncomment{Perhaps add here that there are also versions of the De
  Loera \etal algorithm with added redudant polynomials or alternative
  versions of the Nullstellensatz. Our lower bounds don't apply to
  these versions (but it is not obvious whether they could be extended
  to work).}

One 
limitation of our result is that our hard graphs are 
very specific, and arguably somewhat artificial. 
For the weaker resolution proof system an average-case exponential
lower bound has been shown for
Erd\H{o}s--Rényi random graphs
$\mathcal{G}(n,p)$ where $p$~is slightly above the threshold value
$p_{k}(n)$ at which the graph becomes highly likely to be 
non-$\numcolours$-colourable~\cite{BCCM05RandomGraph}. It is natural to ask
whether these instances are hard for polynomial calculus too.

\begin{openproblem}
  Consider a random graph sampled according to $\mathcal{G}(n,p)$ with
  $p>p_{k}(n)$, so that the graph is non-$\numcolours$-colourable with
  high probability.  Does polynomial calculus require large degree to
  certify non-\kcolourability of such graphs with high probability?
\end{openproblem}
Some progress on this open problem may come from the recent work
of~\cite{RomeroTuncel2022GraphsNullstellensatz} where they show degree
lower bound in Nullstellensatz for large classes of graphs, relying
just on the girth size.

In this paper, we also show that the graph colouring instances that
are provably hard for polynomial calculus are very easy for the
cutting planes proof system. 
It does not seem very likely that graph colouring would be an 
easy problem for cutting planes, however, and so
it would be interesting to find explicit candidates for hard instances
for cutting planes, even if proving the actual lower bounds may be
very hard.  This question is also interesting for the
Lasserre/Sums-of-Squares proof system. Our instances seem likely to be
easy for Lasserre, since they are based on the hardness of the
pigeonhole principle and this combinatorial principle is easy for
Lasserre.

\begin{openproblem}
  Find candidates for explicit hard instances of non-$3$-colourability
  for cutting planes and for Lasserre/Sums-of-squares proof systems,
  and then prove formally that these instances are indeed hard.
\end{openproblem}

A final, intriguing, observation, which is somewhat orthogonal to the
rest of this discussion, is that even though the graph colouring
instances in our paper are easy for cutting planes, results from the
\emph{Pseudo-Boolean Competition 2016} indicate that they are quite
hard in practice for state-of-the-art pseudo-Boolean
solvers~\cite{PB16resultsDECSMALLINT}. This is even more interesting
considering that the cutting planes refutations that we construct have
small rank (i.e., the maximum number of application of the division
rules along any path in the proof graph is small).

\section*{Acknowledgements}

We are grateful to 
Mladen~Mik\v{s}a
and
Alexander~Razborov
for stimulating discussions and helpful feedback during various stages
of this project.   
We would also like to thank Jan~Elffers for running experiments with
pseudo-Boolean solvers on instances obtained from our reduction from
functional pigeonhole principle formulas to graph colouring,
demonstrating that  these formulas are hard in practice.
Last but not least, a big thanks to the anonymous CCC reviewers, who
helped us catch some typos and bugs that really should not have been
there, and whose suggestions helped improve the exposition
considerably.  

Part of this research  was done while \theauthorML
was at KTH Royal Institute of Technology funded by the
European Research Council (ERC) under the European Union's Seventh Framework
Programme \mbox{(FP7/2007--2013) /} ERC grant agreement no.~279611.
Later work at Universitat Polit\`ecnica de Catalunya for
this project has received funding from the European Research Council
(ERC) under the European
Union's Horizon 2020 research and innovation programme (grant
agreement ERC-2014-CoG 648276 AUTAR).
\TheauthorJN was supported by the European Research Council under the
European Union's Seventh Framework Programme \mbox{(FP7/2007--2013) /}
ERC grant agreement no.~279611, by the Swedish Research Council grants
\mbox{621-2012-5645} and \mbox{2016-00782}, and by the Independent
Research Fund Denmark grant \mbox{9040-00389B}.

\newcommand{\etalchar}[1]{$^{#1}$}

\bibliographystyle{alpha}

\end{document}